\documentclass[letterpaper, 12pt, a4paper, 
]{article}

\usepackage[dvipsnames]{xcolor}
\usepackage[setpagesize=false, colorlinks=true, 
linkcolor= Plum, 
citecolor= PineGreen,%
linktocpage=true,%
]{hyperref}
\usepackage[pdftex]{graphicx}

\usepackage{amsmath, amssymb, amsthm} 				  
\usepackage{bm,centernot}
\usepackage[T1]{fontenc} 
\usepackage{newpxtext, newpxmath}
\usepackage[]{eucal} 
\usepackage[scr=rsfso]{mathalpha}

\usepackage[nameinlink]{cleveref}
\usepackage{aliascnt}
\usepackage{booktabs, setspace, sgame,subcaption}
\usepackage[pdftex, margin=1in, 
]{geometry}
\usepackage[]{natbib} 					
\usepackage[nottoc]{tocbibind} 		
\usepackage[inline]{enumitem}			\setlist{nosep} 

\newcommand{\dx}[1]{ \mathsf{d} \left( #1 \right)}			
\newcommand{\hx}[1]{ \mathsf{h} \left( #1 \right)}			\newcommand{\xd}[2]{{\mathsf{x} \left( #1 , #2 \right)}}
\newcommand{\gs}{\mathcal{X}} 										
\newcommand{\COP}{F^\star} 									\newcommand{\cop}{f^\star}
\newcommand{\set}{\mathbb{X}} 								\newcommand{\op}{\mathbf}
\newcommand{\pdom}{\mathscr{P}_D} 							
 						\newcommand{\cupch}{C_H}
\newcommand{\Menus}{2^{\gs}} 								\newcommand{\InMenus}{\subseteq \gs}
\newcommand{\menu}{X} 										\newcommand{\Ac}{\mathrm{Ac}}

\newcommand{\RH}[2]{ R_H \left( \left\{ {#1}_{1}, \dots, {#1}_{#2}\right\} \right)}
\newcommand{\CHop}[1]{C_H \left( \set (\op{#1})\right)}
\newcommand{\RHop}[1]{R_H \left( \set (\op{#1})\right)}

\newcommand{\ordsucc}{\mathord{\succ}}  						
\newcommand{\varsucc}{\vartriangleright}   						\newcommand{\varsucceq}{\trianglerighteq}
\newcommand{\ordvarsucc}{\mathord{\varsucc}} 					
\newcommand{\drp}[2]{\succ_{#1}^{-#2}} 							\newcommand{\drpeq}[2]{\succeq_{#1}^{-#2}}
\newcommand{\orddrp}[2]{\ordsucc_{#1}^{-#2}} 					
\newcommand{\ordvardrp}[2]{\mathord{\rhd}_{#1}^{-#2}}
\newcommand{\vvsucc}{\sqsupset} 								\newcommand{\ordvvsucc}{ \mathord{\sqsupset}}
\newcommand{\vvvsucc}{\Supset} 									

 				\newcommand{\whiteqed}{\hfill $\square$}

\theoremstyle{plain}
				\newtheorem{thm}{Theorem} \newtheorem*{thm*}{Theorem}
								\Crefmultiformat{thm}{#2Theorems~#1#3}{ and~#2#1#3}{, #2#1#3}{ and #2#1#3}
								\Crefrangeformat{thm}{Theorems~#3#1#4--#5#2#6}
				
								\Crefmultiformat{prop}{#2Propositions~#1#3}{ and~#2#1#3}{, #2#1#3}{ and #2#1#3}
								\Crefrangeformat{prop}{Propositions~#3#1#4--#5#2#6}
				\newtheorem*{prop*}{Proposition}
				\newtheorem{lem}{Lemma}	
								\Crefmultiformat{lem}{#2Lemmas~#1#3}{ and~#2#1#3}{, #2#1#3}{ and #2#1#3}
								\Crefrangeformat{lem}{Lemmas~#3#1#4--#5#2#6}
				\newaliascnt{fact}{thm}
				\newtheorem{fact}[fact]{Fact}
								\aliascntresetthe{fact}
								\Crefname{fact}{Fact}{Facts}
								\Crefmultiformat{fact}{#2Facts~#1#3}{ and~#2#1#3}{, #2#1#3}{ and #2#1#3}
								\Crefrangeformat{fact}{Facts~#3#1#4--#5#2#6}
				\newaliascnt{cor}{thm}
				\newtheorem{cor}[cor]{Corollary}
								\aliascntresetthe{cor}
								\Crefname{cor}{Corollary}{Corollaries}
				\newtheorem*{cor*}{Corollary}

\theoremstyle{definition}
				\newtheorem{ex}{Example}		
								\Crefmultiformat{ex}{#2Examples~#1#3}{ and~#2#1#3}{, #2#1#3}{ and #2#1#3}			
								\Crefrangeformat{ex}{Examples~#3#1#4--#5#2#6}
				\newtheorem{definition}{Definition}
								\Crefmultiformat{definition}{#2Definitions~#1#3}{ and~#2#1#3}{, #2#1#3}{ and #2#1#3}			
								\Crefrangeformat{definition}{Definitions~#3#1#4--#5#2#6}

\newcommand{\mon}{dropping monotonicity}

\newcommand{\header}[1]{\begin{center} \textit{#1} \end{center}}

\begin{document}

\title{Strong Observable Substitutability \\ and  the Cumulative Offer Mechanism\footnote{
		This study partially supersedes an earlier working paper titled ``Stability, strategy-proofness, and respect for improvements'' \citep{refall:1401}. 
		We particularly thank Yasunori Okumura for helpful discussions during the series of projects.
		We are also grateful to Akifumi Ishihara, Fuhito Kojima, Shintaro Miura, Takeshi Murooka, Shin Sato, Jan Christoph Schlegel, Tayfun S\"{o}nmez, Kentaro Tomoeda, Yuichi Yamamoto, and M.~Bumin Yenmez for their comments.
		Much of this research was conducted while Hirata was visiting UNSW Sydney; he is grateful for their hospitality. 
		This study is financially supported by JSPS KAKENHI 19K01542 and 23K12445. 
		The usual disclaimer applies.
		}
}
\author{%
Daisuke Hirata\footnote{Hitotsubashi University; d.hirata@r.hit-u.ac.jp} \and
Yusuke Kasuya\footnote{Kobe University; kasuya@econ.kobe-u.ac.jp} 
}
\date{This Version: \today}

\maketitle
\thispagestyle{empty}

\begin{abstract}
We study the properties of the cumulative offer mechanism (COM) when institutions' choice functions satisfy \emph{strong observable substitutability} (strong OS). 
First, we show that each of the following properties of the COM characterizes strong OS: IR monotonicity, weak Maskin monotonicity, and dropping monotonicity.
Dropping monotonicity is a new condition weaker than weak Maskin monotonicity, and it is interpretable as a weakening of strategy-proofness and non-bossiness. 
Second, we show that when choice functions are strongly OS, weak group strategy-proofness of the COM reduces to individual strategy-proofness. 
However, strong OS  is neither necessary nor ``almost necessary" for this reduction.
\bigskip \quad

\end{abstract}

\newpage
\thispagestyle{empty}
\setcounter{tocdepth}{2}
{
\tableofcontents
}

\onehalfspacing
\doublespacing
\newpage
\pagestyle{plain}
\setcounter{page}{1}

		\section{Introduction}

Matching with contracts \citep{refall:641} generalizes the classic model of priority-based matching, such as school choice \citep[e.g.,][]{refall:581}, in two ways. 
First, it allows multiple possible ways for an agent (e.g., student, worker, cadet, etc.) to be matched to an  institution (e.g., school, hospital, military branch, etc.).
Specifically, an allocation in matching with contracts is a set of contracts, each of which specifies an agent and institution to match as well as \emph{how} to match them.\footnote{
			Examples of different terms for an agent-institution pair to be matched include 
				tuition levels at a university \citep{refall:1647, refall:1448}, 
				salary levels and jobs at an employer \citep{refall:640, refall:1150},
				lengths of service at a military branch \citep{refall:1451, refall:1155, refall:1154}, and 
				waiting times for legal traineeships at a regional court \citep{refall:1183}.
			} 
It contrasts with the classic model, where the relation for each agent-institution pair is binary (i.e., either matched or unmatched).
Second, matching with contracts models an institution's priority structure as a \emph{choice function}, which selects a subset of contracts from each possible menu.
This specification also broadens the scope of the model as it can accommodate those priorities that cannot be reduced to a linear ranking over agents.\footnote{%
			A typical example is affirmative actions in school choice and college admissions \citep[e.g.,][]{refall:581,  refall:1253, refall:1382, refall:1454}.
			With minority reserves or majority quota, a school's admission criterion cannot be reduced to a simple ranking over students.
			}

In studies of matching with contracts, \emph{substitutability} and the \emph{cumulative offer mechanism} (henceforth, COM) are central. 
\citet{refall:641} define substitutability as the monotonicity of rejected sets: A choice function is substitutable if any contract unchosen from a smaller menu is never chosen from a larger menu (in the set sense).  
Under this condition, \citet{refall:641} show that various results in the classic model extend to matching with contracts. 
In particular, they show that the COM, which is a variant of the deferred acceptance \citep{refall:1163}, is \emph{stable} and \emph{strategy-proof} 
	when institutions' choice functions satisfy substitutability along with another condition called size-monotonicity (a.k.a.~the law of aggregate demand). 
Since stability and strategy-proofness are two central desiderata in matching market design, 
	 the subsequent literature has generalized this result and broadened its applicability to real-world markets. 
To this end, as we will briefly overview in \Cref{subsec:lit}, 
	several substitutability conditions have been proposed and studied by weakening \citeauthor{refall:641}'s (\citeyear{refall:641}) original. 
	While these conditions are known to have substantially different implications for the COM, the sources of these differences are not yet fully understood.

To better understand different substitutability conditions, 
this paper investigates the implications of \emph{strong observable substitutability} (henceforth, strong OS; \citealp{refall:1401}) and distinguishes it from other related conditions.
As its name suggests, strong OS strengthens \emph{observable substitutability} (henceforth, OS; \citealp{refall:1188}).
Both OS and strong OS are substitutability conditions restricted to the menus that arise during the computation process (i.e., the so-called \emph{cumulative offer process}) of the COM. 
The only difference lies in that OS compares menus only along each single path of the process, whereas strong OS compares across different paths. 
Thus, strong OS is stronger in imposing more restrictions on a choice function than OS. 
At the same time, strong OS is relatively weak among the existing substitutability conditions; see \Cref{fig:rel} below.

Despite their similar definitions, OS and strong OS have different implications for the COM: 
Specifically, OS and size-monotonicity are insufficient for the strategy-proofness of the COM \citep{refall:1188}, whereas strong OS and size-monotonicity jointly suffice \citep{refall:1401}. 
Nevertheless, it remains unclear how and why the two conditions diverge in their economic implications, even though the difference is mathematically apparent in their definitions.

To shed light on the gap between the two, we first study the monotonicity properties of the COM implied by strong OS.
We show that strong OS is characterized by \emph{each} of the following three properties: IR monotonicity, weak Maskin monotonicity, and \mon.
While the first two are from \citet{refall:680}, \mon~is a new concept and is the weakest among the three. 
\Cref{thm:mon} establishes that 
	the COM satisfies all three if choice functions are strongly OS, and it satisfies none of them otherwise. 
This theorem draws a definitive boundary around strongly OS choice functions, and it suggests any results that exploit these monotonicity conditions of the COM would also necessitate strong OS. 
Building upon \Cref{thm:mon}, we also extend \citet{refall:1189}'s characterization: Under strong OS, the COM is the only stable mechanism that satisfies weak Maskin monotonicity (\Cref{thm:uniqueness}). 
In this result, we \emph{cannot} replace weak Maskin monotonicity with dropping monotonicity.
\Cref{thm:uniqueness} thus distinguishes among the three monotonicity properties that are symmetric in \Cref{thm:mon}.

Second, we examine how strong OS relates (individual) strategy-proofness to weak group strategy-proofness for the COM.
Specifically, we show that under strong OS, strategy-proofness for the COM \emph{implies} weak group strategy-proofness (\Cref{thm:GSP}).
When stronger substitutabilities are combined with size-monotonicity, the COM is known to be not only strategy-proof but also weakly group strategy-proof.\footnote{
	See \citet{refall:1395,refall:1007} and \citet{refall:1191}.
	}
\Cref{thm:GSP} partially generalizes these existing results, and as a corollary, we also establish a new sufficient condition for the weak group strategy-proofness of the COM (\Cref{cor:suff}).
It should be noted that OS is insufficient for \Cref{thm:GSP}. 
At the same time, strong OS is \emph{not necessary} for strategy-proofness to imply weak group strategy-proofness, even in the ``almost necessity'' sense (\Cref{thm:anti-nec}).
See \Cref{sec:GSP} for details including the formal statement. 

The rest of the paper is organized as follows:
\Cref{subsec:lit} briefly overviews the related literature. 
\Cref{sec:model} introduces the model and basic definitions. 
\Cref{sec:char,sec:GSP} present the results. 
All the proofs are relegated to \Cref{sec:proofs}.

\subsection{Substitutability Conditions in the Literature} \label{subsec:lit}

This subsection briefly overviews the roles played by substitutability conditions in the theory of matching with contracts.
Since \citet{refall:641}, various substitutability concepts have been defined and studied: Key examples include 
	unilateral and bilateral substitutability \citep{refall:1007}, slot-specific priorities \citep{refall:1158}, substitutable completability \citep{refall:1191}, 
	observable substitutability and observable substitutability across agents \citep{refall:1188}, and strong observable substitutability \citep{refall:1401}.\footnote{%
			See also \citet{refall:711}, \citet{refall:1356}, \citet{refall:1380}, and \citet{refall:1391} for related concepts.
			}
\Cref{fig:rel} summarizes the logical relation among those conditions; e.g., unilateral substitutability implies substitutable completability, which in turn implies strong observable substitutability.\footnote{%
	For the logical implications among substitutability conditions, see also \citet{refall:1148} and \citet{refall:1392}.
	}

\begin{figure}
	\begin{center}
		\includegraphics[width=.99\textwidth]{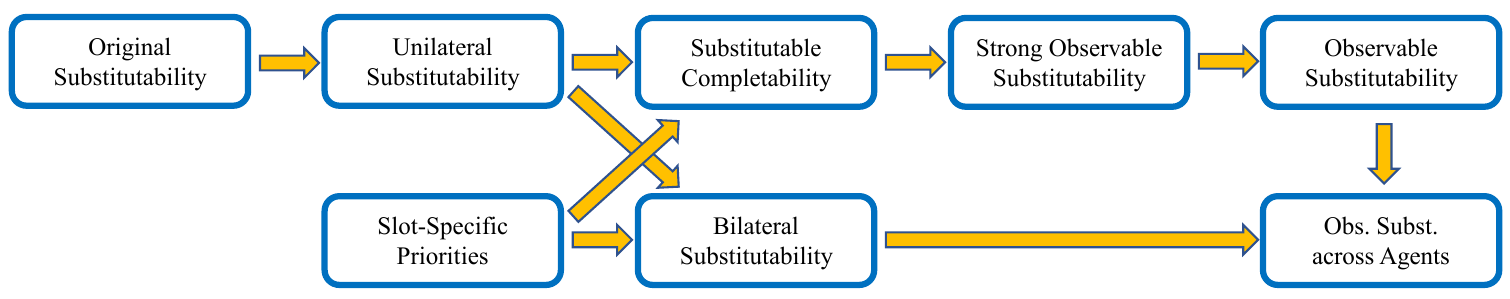}
	\end{center}
	\caption{Logical Relations among Substitutability Conditions} \label{fig:rel}
	\hrulefill
\end{figure}

Roughly speaking, those substitutability conditions have two theoretical roles. 
One is to guarantee the stability of allocations, which is a central desideratum in the matching literature. 
To this end, any of the substitutability conditions we listed at the beginning of this subsection suffices.
In particular, observable substitutability across agents, which is the weakest in the above list, guarantees the stability of the COM outcome for any preference profile of agents.

The other role of substitutability conditions is as a (sub)condition for the strategy-proofness of the COM. 
From this perspective, the above conditions are divided into three categories. 
The first is those that are sufficient for strategy-proofness when combined with size-monotonicity (a.k.a.~the law of aggregate demand). 
As mentioned earlier, \citet{refall:641} first establish such a result with the original substitutability. 
Since then, similar results have been obtained with unilateral substitutability, slot-specific priorities, substitutable completability, and strong observable substitutability. 
See the aforementioned references for details. 
Note that strong OS is the weakest among this first category. 
The second category is those that are not directly relevant for strategy-proofness. 
Namely, bilateral substitutability and observable substitutability across agents are insufficient even if we also assume size-monotonicity.
Indeed, to the best of our knowledge, there is no sufficient condition for the strategy-proofness of the COM that involve either of them as a subcondition. 
Third, observable substitutability constitutes its own category:
On one hand, it does not ensure strategy-proofness even in conjunction with size-monotonicity. 
On the other hand, it becomes sufficient when another condition, called non-manipulability via contractual terms, is satisfied on top of size-monotonicity.

		\section{Preliminaries} \label{sec:model}

This section introduces the matching-with-contracts framework, the cumulative offer mechanism, and the substitutability conditions studied in the paper.
Let $D$ and $H$ be finite sets of \emph{agents} and \emph{institutions}, respectively. 
The finite set of all possible \emph{contracts} is $\gs \subseteq D \times H \times \Theta$ for some finite $\Theta$.
The elements of $\Theta$ are called \emph{contractual terms} and represent different ways for a pair $(d,h) \in D \times H$ to be matched.
For each contract $x \in \gs$, let $\dx{x}$ and $\hx{x}$ be its projections onto $D$ and $H$, 
	 i.e., $x = ( \dx{x}, \hx{x}, \theta)$ for some $\theta \in \Theta$.
In other words, each $x$ is a bilateral contract between agent $\dx{x} \in D$ and institution $\hx{x} \in H$. 
In what follows, $\menu$ denotes a generic subset of $\gs$, and $\dx{\menu}$ and $\hx{\menu}$ denote the sets $\{\dx{x}: x \in \menu\}$ and $\{\hx{x}: x \in \menu\}$, respectively.
The power set of $\gs$ is denoted by $\Menus$.

A subset $\menu \InMenus$ of contracts is said to be an \emph{allocation}
	if it includes at most one contract for each agent, i.e., 	if $x, x' \in \menu$ and $x\neq x'$ imply $\dx{x} \neq \dx{x'}$.
Note that an institution can be assigned multiple contracts at an allocation; this is a many-to-one model. 
The set of all possible allocations is denoted by $\mathscr{A} \subseteq \Menus$. 
For each allocation $\menu \in \mathscr{A}$ and  agent $d \in D$, 
	let $\xd{d}{\menu}$ denote the contract that $\menu$ assigns to $d$; i.e., $\xd{d}{\menu} = x $ if $x \in \menu$ and $\dx{x} =d$.  
If there is no such contract in $\menu$ for agent $d$, then she is said to be assigned the \textit{null contract} and we let $\xd{d}{\menu} = \varnothing$.
In what follows, we use the symbols $\varnothing$ and $\emptyset$ to denote the null contract and the empty set, respectively.

Each agent $d \in D$ has a strict preference relation represented by a linear order $\succ_d$ over $\gs_d \cup \{\varnothing\}$, where $\gs_d := \{ x \in \gs: \dx{x} = d\}$ is the set of contracts relevant for $d$.
A contract $x \in \gs_d$ is said to be \emph{acceptable} if $x \succ_d \varnothing$, and $\Ac (\ordsucc_d) := \{ x \in \gs_d : x \succ_d \varnothing\}$ denotes the set of all acceptable contracts for $\succ_d$.
While preferences are assumed to be strict, we write $x \succeq_d y$ to denote [$x \succ_d y$ or $x = y$]. 
Note that we can naturally extend preferences over contracts to those over allocations:
When $X$ and $Y$ are allocations, with a slight abuse of notation, we write $X \succ_d Y$, $Y =_d z$, and so on to denote $\xd{d}{X} \succ_d \xd{d}{Y}$, $\xd{d}{Y} = z$, and so on. 
A profile of the agents' preferences and the domain of all possible profiles are denoted by
	$\mathord{\succ_D} = (\succ_d)_{d \in D}$ and $\mathord{\mathscr{P}_D} =\prod_{d\in D} \mathscr{P}_d$,
	respectively.

Taking a subset $Y$ of contracts and a preference $\succ_d$ as given, 
	the \emph{dropping} of $Y$ from $\succ_d$, denoted by $\orddrp{d}{Y}$,
	is a preference for the same agent $d$ such that 
	(i) $\Ac \left(\orddrp{d}{Y}\right) = \Ac (\ordsucc_d) - Y$ and (ii) $w \drp{d}{Y} w' \Leftrightarrow w \succ_d w' $ for all $w, w' \in \Ac \left(\orddrp{d}{Y}\right) $.
Strictly speaking, $\orddrp{d}{Y}$ is an equivalence class of preferences 
	because the two conditions impose no restriction on the ranking among unacceptable contracts.\footnote{
		For instance, suppose $\gs_d = \{x,y,z,w\}$ and $\ordsucc_d$ is such that $x \succ_d y \succ_d z \succ_d \varnothing \succ_d w$. 
		Then, two preferences $\ordvarsucc_d$ and $\ordvvsucc_d$ both satisfy the two conditions for $\drp{d}{\{y\}}$, where
			 $x \varsucc_d  z \varsucc_d \varnothing \varsucc_d y \varsucc_d  w$ and $x \vvsucc_d  z \vvsucc_d \varnothing \vvsucc_d w \vvsucc_d y$. 
		}
However, this study exclusively considers a particular mechanism, the cumulative offer mechanism, and its outcome depends only on the preferences over acceptable contracts.
Therefore, we treat it as if it is a single preference, and it should not cause any confusion. 
Note also that $\orddrp{d}{Y}$ is well-defined even if $\dx{y} \neq d$ for some $y \in Y$.
In particular, $\orddrp{d}{Y} = \ordsucc_d$ if $d \not\in \dx{Y}$.
Given a profile $\ordsucc_D = ( \ordsucc_d )_{ d \in D }$, we simply write $\orddrp{D}{Y}$ to denote $\left( \orddrp{d}{Y} \right)_{ d \in D }$.
When the dropped set of contracts is a singleton, for brevity,
	we write $\orddrp{d}{x}$ and $\orddrp{D}{x}$ instead of  $\orddrp{d}{\{x\}}$ and $\orddrp{D}{\{x\}}$, respectively.

A \emph{mechanism} is a mapping $f: \mathscr{P}_D \to \mathscr{A}$ that 
	associates each possible preference profile $\ordsucc_D$ with an allocation $f (\ordsucc_D)$. 
A mechanism $f (\cdot)$ is said to be \emph{(individually) strategy-proof} if there are no $\ordsucc_D, \ordvarsucc_D \in \pdom$ and $d \in D$ such that 
$f (\ordvarsucc_D) \succ_d f (\ordsucc_D)$ and $ \ordsucc_{d'} = \ordvarsucc_{d'} $ for all $d' \in D - \{d\}$.
It is called \emph{non-bossy} if there are no $\ordsucc_D, \ordvarsucc_D \in \pdom$ and $d \in D$ 
	such that $ \ordsucc_{d'} = \ordvarsucc_{d'} $ for all $d' \in D - \{d\}$, $f (\ordsucc_D) =_d f (\ordvarsucc_D)$, and $f (\ordsucc_D) \neq  f (\ordvarsucc_D)$;
	here, as we defined above, $f (\ordsucc_D) =_d f (\ordvarsucc_D)$ means that the two allocations assign agent $d$ the same (possibly null) contract.

	\subsection{Choice Functions and the Cumulative Offer Mechanism}

This subsection introduces institutions’ choice functions and defines the cumulative offer mechanism.
Each institution $h\in H$ has a \emph{choice function} $C_h: \Menus \to \mathscr{A}$ such that 
	for every menu $\menu \InMenus$ of contracts, $C_h (\menu) \subseteq \menu$ and $\hx{x}=h$ for all $x \in C_h (\menu)$.
Throughout the paper, we assume that the choice functions satisfy the following mild requirement: 
Institution $h$'s choice function $C_h(\cdot)$ is said to satisfy 
	the \textit{irrelevance of rejected contracts} (for short, IRC; \citealp{refall:966}) if 
	\begin{align*}
	\left[ x \not\in C_h (\menu \cup \{x\}) \Rightarrow C_h (X \cup {\{ x \}} )= C_h (\menu) \right]
	\mbox{ for all } x \in \gs \mbox{ and } \menu \InMenus.
	\end{align*}
Note that this condition is satisfied when a choice function is induced by a strict preference over subsets of contracts, which is the case in all examples below.
Taking a choice function $C_h$ as given, the \emph{rejection function} $R_h$ \emph{associated with} $C_h$ is defined by $R_h (\menu) := \menu - C_h (\menu)$ for each $\menu \InMenus$.
A profile of the institutions' choice functions is denoted by $C_H= (C_h)_{h \in H}$. 
Slightly abusing notation, we will often identify $C_H$ with the aggregate choice function, so that $C_H (\menu)$ denotes $\cup_{h} C_h (\menu)$ for each $\menu \InMenus$.
Note that the aggregate $C_H (\cdot)$ should satisfy the IRC, given that each component $C_h$ does.
The \emph{aggregate rejection function associated with $C_H$} is defined by  $R_H (\menu) := \menu - C_H (\menu) = \cap_h R_h (\menu)$ for each $\menu \InMenus$.

We call a tuple $\left( D, H, \gs, C_H \right)$ a \emph{market}. 
Taking a market as given, we define \emph{stability} as follows.
An allocation $\menu \in \mathscr{A}$ is said to be \textit{individually rational} at $(\succ_D, C_H)$, if 
	$X \succeq_{d} \varnothing$ for all $d \in D$ and $C_h(\menu)=\{ x \in \menu: \hx{x} = h \}$ for all $h \in H$.
A pair of an institution $h \in H$ and a subset $\menu' \InMenus$ of contracts is said to \textit{block} an allocation $\menu \in \mathscr{A}$ at $(\succ_D, C_H)$ if 
 $C_h ( \menu \cup \menu' )  \neq C_h (\menu)$ and $C_h(\menu\cup \menu') \succeq_{d}\menu$ for all $d \in \dx{ C_h (\menu\cup \menu')}$.\footnote{
			Requiring $C_h (\menu \cup \menu') = \menu'$ is redundant here, although it is often a part of the definition in the literature.
			See \citet[][Lemma 1]{refall:1324} for details.
			}
An allocation $\menu$ is said to be \textit{stable} at $(\succ_D, C_H)$ if it is individually rational and not blocked by any $(h, \menu') \in H \times \Menus$.
A mechanism $f$ is stable if its output $f(\ordsucc_D)$ is stable for any $\ordsucc_D \in \pdom$.

Lastly, taking a market as given and fixed, we algorithmically define the \emph{cumulative offer mechanism} as \Cref{def:COP} below. 
Strictly speaking, the following definition is improper in two ways:
First, it does not fully specify who should make an ``offer'' at each step of the algorithm when multiple agents are qualified. 
Without any assumption on $C_H$, the outcome of the algorithm may vary depending on who is chosen to make an offer. 
Second, the outcome may contain multiple contracts for a single agent; i.e., it may not be an allocation as defined above. 
Nevertheless, we will impose an assumption on $C_H$ in the next subsection so that this mechanism is uniquely defined and always outputs an allocation. 

\begin{definition} \label{def:COP}
Given a market $\left( D, H, \gs, C_H \right)$ as well as a preference profile $\ordsucc_D \in \pdom$, 
	the \emph{cumulative offer process} (for short, \emph{COP}) computes a subset of contracts as follows.
	\begin{itemize}
	\item Initial condition: Let $D_0 = D$ and $P_0 = \emptyset$. 			
	\item Step $t \geq 1$:
				Arbitrarily fix an agent $d_t \in D_{t-1}$. 
				She offers her best contract,  say ${x_t}$, among those remaining (i.e., among $\gs - P_{t-1}$) according to her preference $\ordsucc_{d_t}$. 
				Let $P_t = P_{t-1} \cup \{ {x_t} \}$ be 
					the pool of contracts that have been offered up to this step.
				Among $P_t$, each institution $h$ holds the best combination of contracts, $C_h (P_t)$.
Lastly, let $D_t$ be the set of agents who are currently unmatched and still have acceptable contracts to offer; i.e.,
				\begin{align*} 
							D_t = \{ d \in D:  
							d \not\in \dx{C_H (P_t)} \mbox{ and } 	\Ac (\succ_d) - P_t \neq \emptyset \}. 
				\end{align*}
				Proceed to step $t+1$ if $D_t$ is non-empty, and terminate otherwise. 
	\item Outcome: 
				When the process terminates after step $T$, its outcome is $C_H (P_{T})$. 	
	\end{itemize} 
The \emph{cumulative offer mechanism} (for short, \emph{COM}) is the mechanism $\cop: \pdom \to \mathscr{A}$ that 
	assigns the outcome of the above process to each $\ordsucc_D$.
\whiteqed
\end{definition}

	\subsection{Substitutability Conditions}

We next introduce three substitutability conditions for choice functions, but to do so, we need several preliminary definitions: 
An \emph{offer process} is a finite sequence $\op{x} = (x_1, \dots, x_n)$ of distinct contracts,
	and its range as a set (rather than a sequence) is denoted by $\set (\op{x}) := \{ x_1, \dots, x_n\}$.
Roughly speaking, an offer process is \emph{observable at $C_H = (C_h)_{h\in H}$} if it can arise during the COP with \emph{some} preference profile and the fixed choice function profile $C_H$. 
More formally, $(x_1,\dots, x_n)$ is observable at $C_H$ if 
	\begin{align} \label{eq:obs}
		\dx{x_{t+1}} \not\in \dx{ \Bigl. C_H (\{x_1, \dots, x_t\})}
		\mbox { for each } t \in \{1,\dots, n-1\}.
	\end{align}
Note that \cref{eq:obs} means agent $\dx{x_{t+1}}$ does not hold a contract at $C_H (\{x_1, \dots, x_t\})$;
	as such, she is eligible to make an offer at step $t+1$, and $x_{t+1}$ can indeed be her offer if it is her next acceptable contract.

We are now ready to define the three substitutability conditions as follows.
While we could define these conditions for individual choice functions, we define them for profiles of them for the sake of brevity. 
One can easily check the definitions for individual functions are equivalent to those for profiles.\footnote{
		More precisely, a profile $C_H = (C_h)_{h\in H}$ meets the following conditions if and only if 
			each component $C_h$ meets the corresponding conditions.
		}

\begin{definition} \label{def:subst}
A profile $C_H$ of choice functions is said to be \emph{strongly observably substitutable} (for short, strongly OS) if 
		\begin{align*}
		R_H \left(\set (\op{x}) \right) \subseteq R_H  \left(\set (\op{y}) \right),
		\end{align*}
		for any  two observable offer processes $\op{x} = (x_1, \dots, x_m)$ and $\op{y} = (y_1, \dots, y_n)$ such that $\set (\op{x}) \subseteq \set (\op{y})$.
It is said to be \emph{observably substitutable} (for short, OS) if 
	\begin{align*} R_H (\{ x_1, \dots, x_{n-1} \}) \subseteq R_H (\{x_1, \dots, x_{n}\}), \end{align*} for any observable offer process $(x_1, \dots, x_n)$.
It is said to be \emph{observably substitutable across agents} (for short, \emph{OSaA}) if 
	\begin{align*} 
			\left[x \in R_H \left( \{x_1, \dots, x_{n-1}\}\right) - 
			{	R_H \left( \{x_1, \dots, x_{n}\}\right) }
			\Bigr.\right]  \Rightarrow \left[ \dx{x} \in \dx{C_H \left(\{x_1, \dots, x_{n-1}\}\right) \Bigr.} \right],
	\end{align*}
	holds for any observable offer process $(x_1, \dots, x_n)$.
\whiteqed
\end{definition}

Comparing the three conditions, strong OS implies OS, which in turn implies OSaA. 
First, strong OS is logically stronger than OS in that 
	the former compares $\op{x}$ and $\op{y}$ arising from different paths of COPs, while the latter only considers the case where $\op{x}$ is a subprocess of $\op{y}$.
Second, OS is stronger than OSaA in that 
	the former precludes any instance of $x \in R_H \left( \{x_1, \dots, x_{n-1}\}\right) - {	R_H \left( \{x_1, \dots, x_{n}\}\right) }$
	whereas the latter admits it as long as $ \dx{x} \in \dx{C_H \left(\{x_1, \dots, x_{n-1}\}\right)}$.
\Crefrange{ex:not-sOS}{ex:not-OS} below demonstrate that these logical relations are strict.
As we explained in the introduction, it should also be noted that  
	all of the three conditions are strictly weaker than \citeauthor{refall:641}'s (\citeyear{refall:641}) original substitutability, which requires $R_H (X) \subseteq R_H (Y)$ for any $X \subseteq Y \subseteq \gs$, and than several other conditions. 
Actually, it is weaker than substitutable completability \citep{refall:1191}, and hence, it is also weaker than unilateral substitutability \citep{refall:1007} and slot-specific priorities \citep{refall:1158};
	see \citet{refall:1401} for a proof.

\begin{ex}[Strong OS is strictly stronger than OS] \label{ex:not-sOS}
Let $D = \{d_1, d_2, d_3\}$, $H = \{h\}$, and $\gs= \{x_1, x_2, y_2, x_3, y_3\}$, where $\dx{x_1} = d_1$, $\dx{x_2} = \dx{y_2} =d_2$, and $\dx{x_3}=\dx{y_3} = d_3$. 
Suppose that $C_h$ is the choice function induced by the following preference relation $\succ_h$ over the subsets of contracts: 
	\begin{align*}
		\{x_1, y_2, y_3\} 	&\succ_h \{x_1, y_3\} \succ_h \{ y_2, x_3\} \succ_h \{x_1,x_2\} \\
								&\succ_h [\mbox{all the other feasible doubletons}] \succ_h [\mbox{all the singletons}] \succ_h \emptyset, 
	\end{align*}
	where any tripleton except $\{x_1,y_2, y_3\}$ is unacceptable.
The rankings among ``the other'' feasible doubletons and those among the singletons are arbitrary. 
In what follows, we confirm that $C_h$ is OS but not strongly OS.

Let $\op{w}^3 = (w_1, \dots, w_3)$ be an observable offer process at $C_h$. 
Then, as $C_h$ accepts any feasible doubleton, the three contracts should involve the three distinct agents; i.e., we have only four cases for $\set \left( \op{w}^3\right)$ as listed in \Cref{table:not-sOS}.
Suppose, for instance, $\set \left( \op{w}^3\right) = \{x_1, x_2, x_3\}$. 
Then, according to $\ordsucc_h$ above, $C_h$ chooses $\{x_1, x_2\}$ and rejects $x_3$. 
Therefore, $\op{w}^4 = (w_1, \dots, w_4)$ is observable only if $w_4 = y_3$.
Moreover, since $C_h$ chooses $\{x_1, y_3\}$ and rejects $\{x_2, x_3\}$ from $\set \left( \op{w}^3\right) = \{x_1, x_2, x_3, y_3\}$, we must have $w_5 = y_2$ for $\op{w}^5 = (w_1, \dots, w_5)$ to be observable. 
As $C_h$ rejects $\{x_2, x_3\}$ again from $\{x_1, x_2, x_3, y_3, y_2\}$, the rejected set monotonically increases along this path of the offer processes. 
With \Cref{table:not-sOS}, one can easily confirm that no violation of OS arises along the other three cases, either. 
That is, $C_h$ is an OS choice function. 
However, $R_h (\{x_1, y_2, x_3\}) = \{x_1\} \not\subseteq \{x_2, x_3\} = R_h ( \{x_1, x_2, x_3, y_3, y_2\})$,
	even though both $(x_1, y_2, x_3)$ and $(x_1, x_2, x_3, y_3, y_2)$ are observable.
Hence, $C_h$ is \emph{not} strongly OS even though it is OS. 
\whiteqed
\end{ex}

\begin{table}[t!]
	\renewcommand{\arraystretch}{1.5}
	\begin{center}
	\begin{tabular}{ccccccc} \toprule[1.5pt]
	$\set \left( \op{w}^3 \right)$ 	&$R_h \left(\set \left( \op{w}^3 \right)\right)$  
	&$w_4$ &$R_h \left(\set \left( \op{w}^4 \right)\right)$ &$w_5$ 	&$R_h \left(\set \left( \op{w}^5 \right)\right)$  \\ \midrule
	$\{x_1, x_2, x_3\}$ 			&$\{x_3\}$ 											&$y_3$ 		&$\{x_2, x_3\}$		&$y_2$ 	&{\color[named]{Maroon}$\pmb{\{x_2, x_3\}}$}\\
	$\{x_1, x_2, y_3\}$ 			&$\{x_2\}$ 											&$y_2$		&$\{x_2\}$ \\
	$\{x_{1}, y_{2}, x_{3}\}$ 	&{\color[named]{Maroon}$\pmb{\{x_1\}}$} 		\\ 
	$\{x_{1}, y_{2}, y_{3}\}$ 	&$\emptyset$ \\ 
	\bottomrule[1.5pt]
	\end{tabular} 
	\end{center}
	\caption{Observable offer processes for $\succ_h$ in \Cref{ex:not-sOS}.} \label{table:not-sOS}
	\rule{\textwidth}{.1pt}
\end{table}

\begin{ex}[OS is strictly stronger than OSaA] \label{ex:not-OS}
Let $D = \{d_1, d_2\}$, $H = \{h\}$, and $\gs = \{x_1, x_2, y_1, y_2\}$, where for each $i \in \{1,2\}$, $x_i$ and $y_i$ are two different contracts between $d_i$ and $h$.
Suppose that $C_h$ is the choice function induced by the following preference relation $\succ_h$ over the subsets of contracts: 
	\begin{align*} \{x_1, x_2\} \succ_h \{y_1\} \succ_h \{y_2\} \succ_h \{x_1\} \succ_h \{x_2\} \succ_h \emptyset, \end{align*}
	where all the doubletons except $\{x_1, x_2\}$ are unacceptable. 
This choice function violates OS	but it does so only along two observable offer processes: $(w_1, \dots, w_4) = (x_1, y_2, y_1, x_2)$ and $(y_2, x_1, y_1, x_2)$.
Along these processes, $C_h$ chooses $x_1$ from $\{x_1, y_2, y_1, x_2\}$ while rejecting it from $\{x_1, y_2, y_1\}$.
However, it chooses $y_1$ for $d_1 = \dx{x_1}$ from $\{x_1, y_2, y_1\}$.
Therefore, $C_h$ satisfies OSaA although it violates OS. 
\whiteqed
\end{ex}

Throughout the rest of this paper, we assume that the profile of choice functions meets OSaA. 
This implies that the COM is uniquely defined and its output $\cop (\ordsucc_D)$ is always an allocation for any $\ordsucc_D \in \pdom$;
	moreover, OSaA is ``almost'' necessary and sufficient for the stability of the COM \citep{refall:1188}, which is a central desideratum in the literature.
Without OSaA, put differently, the COM may not be well-defined, and even if it is, it loses its key property.\footnote{
			See also \citet{refall:1007}, \citet{refall:1147}, and \citet{refall:1158} for the stability and uniqueness of the COM.
			}
Hence, the assumption of OSaA would be almost indispensable in studying the COM.

		\section{Characterizations of Strong OS} \label{sec:char}

In this section, we characterize strongly OS choice functions based on monotonicity properties of the COM. 
We introduce three properties and show that the COM meets any of them if and only if the profile of choice functions is strongly OS. 
Since each of the three is defined based on a class of preference transformations, we introduce those transformations first.

\begin{definition}  \label{def:transformations}
A preference $\ordsucc_d \in \mathscr{P}_d$ of agent $d$ is said to be
	an \emph{IR monotonic transformation} of another preference $\ordvarsucc_d \in \mathscr{P}_d$ at $x \in \gs \cup \{\varnothing\}$ if 
		\begin{align} \label{eq:def:IR}
				\{ w \in \Ac (\ordsucc_d): w \succ_d x \} \subseteq \{ w \in \gs \cup \{ \varnothing\}: w \varsucc_d x \}.
		\end{align}
It is said to be a \emph{monotonic transformation} of $\ordvarsucc_d $ at $x \in \gs \cup \{\varnothing\}$ if 
		\begin{align} \label{eq:def:weak-Maskin}
				\left\{ w \in \gs \cup \{ \varnothing\}: w \succ_d x \right\} \subseteq \left\{ w \in \gs \cup \{ \varnothing\}: w \varsucc_d x \right\}.
		\end{align}
It is said to be a \emph{dropping transformation} of $\ordvarsucc_d $ at $x \in \gs \cup \{\varnothing\}$ 
	if $\ordsucc_d = \ordvardrp{d}{Y}$ for some $Y \subseteq \gs$ such that $x \not\in Y$.
\whiteqed
\end{definition}

Comparing the three, IR monotonic transformations are a superset of monotonic transformations, which in turn are a superset of dropping transformations.
First, comparing \cref{eq:def:weak-Maskin,eq:def:IR}, 
	the left-hand side for an IR monotonic transformation, $\{ w \in \Ac (\ordsucc_d): w \succ_d x \}$,
	is a subset of its counterpart for a monotonic transformation, $\{ w \in \gs \cup \{ \varnothing\}: w \succ_d x \}$.
Hence,	a monotonic transformation of $\varsucc_d$ at $x$ is always an IR monotonic transformation of $\varsucc_d$ at $x$.
Second, if $\ordsucc_d = \ordvardrp{d}{Y}$ for some $Y \not\ni x$, we have 
	\begin{align*}
		\left\{ w \in \gs \cup \{ \varnothing\}: w \succ_d x \right\} = \left\{ w \in \gs \cup \{ \varnothing\}: w \varsucc_d x \right\} - Y. 
	\end{align*}
Thus, a dropping transformation of $\varsucc_d$ at $x$ is always a monotonic transformation of $\varsucc_d$ at $x$. 
Furthermore, these logical relations among the three are strict as the following example demonstrates.

\begin{ex} \label{ex:transformations}
Let $\gs_d = \{x,y,z\}$ and define three preferences over $\gs_d \cup \{\varnothing\}$ as follows: 
	$y \varsucc_d x \varsucc_d \varnothing \varsucc_d z$,
	$\varnothing \succ_d x \succ_d y \succ_d z$, and 
	$x \vvsucc_d y \vvsucc_d \varnothing \vvsucc_d z$. 
Comparing these three, we make two observations:
First, $\ordsucc_d$ is an IR monotonic transformation of $\varsucc_d$ at $x$, while it is not a monotonic transformation.
Second, $\ordvvsucc_d$ is a monotonic transformation of $\varsucc_d$ at $x$, whereas it is not a dropping transformation. 
\whiteqed
\end{ex}

Based on the three transformations, we define three monotonicity properties for a mechanism. 
The first two are \emph{IR monotonicity} and \emph{weak Maskin monotonicity}.
In the classic matching without contracts, \citet{refall:680} introduce these concepts as a part of the axioms for their characterization of the deferred acceptance mechanism.\footnote{See also \citet{refall:1449} for a related characterization.} 
The third property, \emph{dropping monotonicity}, is newly defined in the present study.

\begin{definition}  \label{def:mon's}
A mechanism $f (\cdot)$ is said to be \emph{IR monotonic} (resp.~\emph{weakly Maskin monotonic} and \emph{dropping monotonic}) if
	\begin{align} f (\ordsucc_D) \succeq_d f (\ordvarsucc_D) \mbox{ for all } d \in D \end{align}
	holds whenever $\ordsucc_D, \ordvarsucc_D \in \pdom$ are such that for each $d \in D$, 
	$\succ_d$ is an IR monotonic transformation (resp.~monotonic transformation and dropping transformation) of $\varsucc_d$ at $\xd{d}{f (\ordvarsucc_D)}$, i.e., at the contract assigned to $d$ under $f(\ordvarsucc_D)$.
\whiteqed
\end{definition}

Notice that because of the aforementioned relations among transformations, IR monotonicity implies weak Maskin monotonicity, which in turn implies dropping monotonicity.
In general, the gaps among the three monotonicities are non-empty:
For instance, the top trading cycles mechanism in the classic model without contracts is not IR monotonic  \citep{refall:680, refall:1450}.
For the COM, however, \Cref{thm:mon} below shows that they become equivalent: 
The COM meets any of the three if and only if it meets all, if and only if the choice functions are strongly OS.

\begin{thm} \label{thm:mon}
	Let $C_H$ be an OSaA profile of choice functions.
	Then, the following are all equivalent: 
	(i) the COM is IR monotonic, (ii) it is weakly Maskin monotonic, (iii) it satisfies \mon, and (iv) $C_H$ is strongly OS.
\end{thm}		
\begin{proof} See \Cref{subsec:proof:COM-mon}. \end{proof}

Under unilateral substitutability \citep{refall:1007}, 
	a stable mechanism satisfies weak Maskin monotonicity if and only if it is the COM \citep{refall:1189}.
\Cref{thm:mon} generalizes the ``if'' part of this characterization by weakening unilateral substitutability to strong OS.
We next show that the ``only if'' part can also be extended, yielding a full characterization under strong OS as follows.

\begin{thm} \label{thm:uniqueness}
Let $C_H$ be a strongly OS profile of choice functions. 
Then, a mechanism is stable and weakly Maskin monotonic if and only if it is the COM.
\end{thm}
\begin{proof}See \Cref{subsec:proof:uniquness}.\end{proof}

A couple of remarks are in order regarding \Cref{thm:uniqueness}:
First, we cannot replace weak Maskin monotonicity with another property called \emph{truncation-proofness}.
Under unilateral substitutability, this property (along with stability) also characterizes the COM \citep{refall:1189}.
Under strong OS, this is not the case because a non-COM stable mechanism can be truncation-proof.
Second, we cannot replace weak Maskin monotonicity with dropping monotonicity, either.
This is because the former is strictly stronger than the latter for general mechanisms, although they are equivalent for the COM.
As a consequence, a non-COM stable mechanism can satisfy dropping monotonicity while violating weak Maskin monotonicity.
See \Cref{sec:remarks} for more details on the subtleties around \Cref{thm:uniqueness}.

\medskip 

To conclude this section, we argue that 
	the three monotonicity properties can be interpreted as weakenings of strategy-proofness and non-bossiness. 
To fix ideas, consider dropping monotonicity, which can be rewritten as
		\begin{align} \label{eq:alt-def:dropping}
			\left[ Z \cap f \left(\ordvarsucc_D \right) = \emptyset \right] \Rightarrow 
			\left[ f \left(\ordvardrp{D}{Z}\right) \varsucceq_d f (\ordvarsucc_D) \mbox{ for all } d \in D \right].
		\end{align}
Suppose, without loss of generality, that $\dx{Z} = \{d^*\}$ for some agent $d^*$ 
	so that the antecedent then describes a unilateral deviation in which agent $d^*$ drops contracts that are not chosen under $\ordvarsucc_D$.
Then, the strategy-proofness of $f (\cdot) $ requires the consequent of \cref{eq:alt-def:dropping} with $d=d^*$;
	if it fails to hold, agent $d^*$ has an incentive to report $\ordsucc_{d^*}$ when her true preference is $\orddrp{d^*}{Z}$.
Indeed, combined with the symmetric argument for the reverse deviation, strategy-proofness requires $f \left(\ordvardrp{D}{Z}\right) =_{d^*} f(\ordvarsucc_{D})$.

In a strategy-proof mechanism, such deviations therefore do not affect the manipulator’s own assignment. 
Then, non-bossiness requires that the assignments of other agents remain unchanged. 
In contrast, the consequent of \cref{eq:alt-def:dropping} allows some agent $d \neq d^*$ to be strictly better off.
Thus, strategy-proofness and non-bossiness jointly imply dropping monotonicity, and the latter can be seen as a weakening of these two properties. 
Similar arguments apply to IR monotonicity and weak Maskin monotonicity, which impose stronger restrictions by considering a broader class of preference transformations.

		\section{Strong OS and Weak Group Strategy-Proofness} \label{sec:GSP}

In this section, we turn to the \emph{group strategy-proofness} of the COM. 
From the existing literature, either on matching or broader mechanism design, 
	it is known that Maskin monotonicity and non-bossiness are often essential for a mechanism to be (weakly) group strategy-proof
	\citep[e.g.,][]{refall:1587, refall:584, refall:1173,refall:1323,refall:1416}. 
In the previous section, strong OS is characterized by weak Maskin monotonicity. 
We have also argued that the monotonicity properties in \Cref{thm:mon} can be seen as a weakened combination of strategy-proofness and non-bossiness. 
Taken together, these observations would naturally lead us to investigate the implications of strong OS for the group strategy-proofness of the COM.
While there are two standard definitions for group strategy-proofness, we adopt the weaker one as we present below. 
It should be noted that even in the classic model without contracts,
	the COM (or equivalently, the deferred acceptance) meets the stronger version only in highly special cases \citep{refall:1250}.

\begin{definition} \label{def:GSP}
A mechanism $f: \pdom \to \mathscr{A}$ is \emph{weakly group strategy-proof} 
	if there are no $\ordsucc_D, \ordvarsucc_D \in \pdom$ such that 
	$f (\ordvarsucc_D) \succ_d f (\ordsucc_D)$ for all $d \in \{d' \in D: \ordsucc_{d'} \neq \ordvarsucc_{d'} \}$.
\whiteqed
\end{definition}

Our next theorem states that weak group strategy-proofness reduces to individual strategy-proofness for the COM when the choice functions are strongly OS.
The existing literature has identified several sufficient conditions for the COM to be weakly group strategy-proof \citep{refall:1182,refall:1395,refall:1007,refall:1191}.
In contrast, only one condition, the one by \citet{refall:1188}, is known to suffice for strategy-proofness but not for weak group strategy-proofness \citep{refall:1437}.
The following theorem suggests why this is the case: 
	a condition that suffices for strategy-proofness but not for weak group strategy-proofness can arise 
	only when choice functions violate strong OS.\footnote{%
			See also \citet{refall:1326} for the relation between individual and weak group strategy-proofness in a general environment beyond matching markets.
			}

\begin{thm} \label{thm:GSP}
	Suppose that $C_H$ is a strongly OS profile of choice functions.
	Then, the COM  is weakly group strategy-proof if and only if it is strategy-proof.\footnote{
		Since the COM is the unique candidate for a stable and strategy-proof mechanism when $C_H$ is OS \citep{refall:1188},
			we can rephrase the conclusion as follows: a stable mechanism $f (\cdot)$ is weakly group strategy-proof if and only if it is strategy-proof.
	}
\end{thm}
\begin{proof} See \Cref{subsec:proof:GSP}. \end{proof}

As a corollary, we obtain a new sufficient condition for the COM to be weakly group strategy-proof, thereby generalizing the aforementioned results in the literature. 
\citet[Theorem 7]{refall:1401} establish that the COM is strategy-proof if the choice functions meet \emph{strong observable size-monotonicity}, to be defined below, in addition to strong OS. 
Combined with \Cref{thm:GSP} above, we can conclude the same condition is also sufficient for weak group strategy-proofness.

\begin{definition} \label{def:sOSM}
	A profile $C_H $ of choice functions is \emph{strongly observably size-monotonic} (for short, \emph{strongly OSM}) if 
		for any two observable offer processes $\op{x}$ and $\op{y}$ at $C_H$, 
		$ \set (\op{x}) \subseteq \set (\op{y})$ implies $ \# C_H (\set (\op{x})) \leq  \# C_H (\set (\op{y}))$.
\whiteqed
\end{definition}

\begin{thm} \label{cor:suff} 
The COM is weakly group strategy-proof if 
	the profile $C_H$ of choice functions is strongly OS and strongly OSM. 
\end{thm}
\begin{proof} See \Cref{subsec:proof:suff}. \end{proof}

\Cref{thm:GSP} above establishes the sufficiency of strong OS for the equivalence between individual and weak group strategy-proofness. 
At the same time, OS is known to be insufficient for this equivalence; i.e., there exists a profile of choice functions such that COM is strategy-proof but not weakly group strategy-proof \citep{refall:1437}. 
It would then be natural to ask if strong OS is necessary for the equivalence. 
The answer to this question turns out to be negative.
Indeed, strong OS is not even \emph{``almost'' necessary} for weak group strategy-proofness to boil down to individual strategy-proofness. 
To be more precise, we introduce the following definition:

\begin{definition}
Let $\mathcal{M}^0 = \left( D^0, \{h\}, X^0, C_h^0 \right)$ be a single-institution market. 
A market $\mathcal{M} = \left( D, H, \gs, C_H \right)$ is said to be an \emph{extension} of $\mathcal{M}^0$ if 
\begin{itemize}
	\item 	$D \supset D^0$ and $H \ni h$, 
	\item 	$\{ x \in \gs: \hx{x} = h\} = X^0$, and 
	\item 	$C_H = (C_i)_{i \in H}$ meets $C_h = C^0_h$. 
\end{itemize} 
An extension $\mathcal{M}$ of $\mathcal{M}^0$ is said to be \emph{regular} if for each $i \in H-\{h\}$, its choice function $C_i (\cdot)$ meets strong OS and strong OSM.	\whiteqed
\end{definition}

Following the convention in the literature, we would say that strong OS is ``almost necessary'' for a formal proposition if the following is true:
For any single-institution market where the choice function $C_h^0$ violates strong OS, there exists its regular extension such that the proposition is false. 
Notice that each additional $C_i (\cdot)$ is strong OS and strong OSM in our definition of a regular extension,
	whereas the literature typically requires unit-demand choice functions \citep[e.g.,][]{refall:641,refall:711,refall:1188}. 
This difference makes strong OS \emph{easier} to be ``almost necessary'' under our definition, because more extensions qualify as regular. 
In turn, it makes our result \emph{stronger} because it negates ``almost necessity.''

Actually, \Cref{thm:anti-nec} below is even stronger than the simple negation of ``almost necessity'' for the equivalence between individual and weak group strategy-proofness.
Note that the equivalence statement becomes false if and only if the COM \emph{is} strategy-proof but not weakly group strategy-proof. 
Therefore, a choice function is a counterexample against ``almost necessity'' if 
	it is not strongly OS and for any regular extension, either (i) the COM is not strategy-proof or (ii) it is weakly group strategy-proof. 
\Cref{thm:anti-nec} is stronger in constructing a counterexample where the second case always holds. 
In other words, strong OS is \emph{not ``almost necessary'' for weak group strategy-proofness.}

\begin{thm} \label{thm:anti-nec}
Strong OS is not ``almost necessary'' for the weak group strategy-proofness of the COM in the following sense:
	There exists a single-institution market $\mathcal{M}^0$ 
	such that (i) the sole institution's choice function is not strongly OS and (ii) the COM is weakly group strategy-proof in any regular extension $\mathcal{M}$ of $\mathcal{M}^0$. 
\end{thm}
\begin{proof} See \Cref{subsec:proof:anti-nec}. \end{proof}

		\section{Proofs} \label{sec:proofs}

In this section, we provide the proofs of \Crefrange{thm:mon}{thm:anti-nec}.
Before doing so, we introduce one more definition and three lemmas. 
Given a profile $C_H$ of choice functions, 
	an observable offer process $\op{x}$ is \emph{complete at a preference profile $\ordsucc_D$} if it is final in the COP with $(\ordsucc_D, C_H)$.
More formally, $\op{x} $ is complete at $\ordsucc_D$ if it is observable at $C_H$ and 
	\begin{align*}
		\Bigl[ d \in \dx{C_H (\set (\op{x}))} \mbox{ or } \Ac (\ordsucc_d) \subseteq \set (\op{x}) \Bigr] \mbox{ for all } d \in D. 
	\end{align*}
We will exploit the following lemmas through the main proofs.
While they are straightforward, we provide the proofs of these lemmas for the reader's sake in \Cref{sec:omitted_proofs}.

\begin{lem} \label{lem:sOS}
	Suppose that $C_H$ is a strongly OS profile of choice functions.
	Let $\op{x}$ and $\op{y}$ be complete offer processes at $(\ordsucc_D, C_H)$ and $(\ordvarsucc_D, C_H)$, respectively.
	Suppose further that $\Delta_R := R_H (\set(\op{x})) - R_H (\set(\op{y}))$ is non-empty.
	Then, there exists $x^* \in \set (\op{x}) - \set(\op{y})$ such that either 
			[1] $x^* \not\in  \Ac \left( \ordvarsucc_{\dx{x^*}} \right)$ or 
			[2] $y^* \not\succ_{\dx{x^*}} x^*$ and $y^* \varsucc_{\dx{x^*}} x^*$,
			where $y^*$ is the (non-null) contract $\dx{x^*}$ signs at $\cupch (\set (\op{y}))$. 
\end{lem}

\begin{lem} \label{lem:weak:COP}
	Suppose that $C_H$ is a profile of OS choice functions. 
	For any preference profile $\ordvarsucc_D \in \pdom$ and any non-null contract $w \in \gs$, the following holds:
		\begin{enumerate}[label={(\alph*)}]
		\item \label{sublem:weak:COP:unpref}
						if $\dx{w}$ prefers $\cop( \ordvarsucc_D)$ to $w$,  
						then $\cop \left(\ordvardrp{D}{w} \right) = \cop ( \ordvarsucc_D)$; 
		\item 	\label{sublem:weak:COP:chosen}
						if $w$ is chosen at $\ordvarsucc_D$ (i.e., if $w \in \cop ( \ordvarsucc_D)$), 
						then $\cop ( \ordvarsucc_D) \varsucc_{\dx{w}} \cop (\ordvardrp{D}{w})$; and 
		\item 	\label{sublem:weak:COP:least}
						if $w$ is the worst acceptable contract for $\dx{w}$, 
						i.e., if $w' \varsucceq_{\dx{w}} w$ for all $w' \in \Ac \left( \ordvarsucc_{\dx{w}} \right) $, 
						then $\cop \left( \ordvardrp{D}{w} \right) \varsucceq_d \cop (\ordvarsucc_D)$ for all $d \neq \dx{w}$. 
		\end{enumerate} 
\end{lem}

\begin{lem} \label{lem:non-bossy:OS+SP}
	Let $C_H$ be an OS profile of choice functions, and suppose that $\cop(\cdot) $ is strategy-proof.
	For any $\ordsucc_D, \ordvarsucc_D \in \pdom$, then, 
	$ \cop (\ordsucc_D) = \cop (\ordvarsucc_D)$ holds if there is $d \in D$ such that 
		\begin{itemize}
		\item $\cop (\ordsucc_D) \varsucceq_d \varnothing$, 
		\item $\{ x \in \gs: x \varsucc_d \cop (\ordsucc_D) \} = \{ x \in \gs: x \succ_d \cop (\ordsucc_D) \}$, and 
		\item  $\ordvarsucc_{d'} =\ordsucc_{d'} $ for all $d' \in D - \{d\}$.
		\end{itemize} 
\end{lem}

\subsection{Proof of \Cref{thm:mon}} \label{subsec:proof:COM-mon}

\subsubsection{Proof of (i) $\bm{\Rightarrow}$ (ii) $\bm{\Rightarrow}$ (iii)}
	As we argued earlier, this is immediate from the definitions of the three properties. 
\qed


\subsubsection{Proof of (iii) $\bm\Rightarrow$ (iv)}
Suppose that the COM meets \mon. 
Towards a contradiction, suppose also that $C_H$ is not strongly OS; 
	i.e., there are two observable offer processes $\op{x} = (x_1, \dots, x_m)$ and $\op{y} = (y_1, \dots, y_n)$ such that
	$\set (\op{x}) \subseteq \set (\op{y})$ and $R_H(\set (\op{x})) - R_H(\set (\op{y}))$ is non-empty.
Without loss of generality, assume that the length of $\op{x}$ we will consider is the shortest among any such pair of offer processes.  
In particular, this assumption entails for each $m' \in \{1, \dots, m-1\}$,
	\begin{align} \label{eq:shortest}
		\RH{x}{m'} \subseteq \Bigl[ \RH{x}{m'+1} \cap \RHop{y} \Bigr].
	\end{align}
To see this, suppose $\RH{x}{m'}$ is not a subset of $\RHop{y}$ for some $m' < m$.
Then, $\op{x}' := ( x_1, \dots, x_{m'})$ is shorter than the original $\op{x}$, $\set (\op{x}') \subseteq \set (\op{y})$, and $R_H(\set (\op{x}')) - R_H(\set (\op{y}))$ is non-empty; this contradicts the assumption of $\op{x}$ being the shortest.
Similarly, if $\RH{x}{m'} \not\subseteq \RH{x}{m'+1}$,
				$\op{x}' := (x_1, \dots, x_{m'} )$ together with $\op{y}' := (x_1, \dots, x_{m'+1} )$ should be a shorter violation of strong OS.
Therefore, we should have \cref{eq:shortest}.

Through the rest of the proof, arbitrarily fix an index $k \leq m$ such that $x_k \in \RHop{x} - \RHop{y}$.
By the shortest-length assumption, 	agent $\dx{x_k}$ signs no non-null contract at $\CHop{x}$ for the following reasons:
	If there is $x_{k'} \in \CHop{x}$ such that $\dx{x_{k'}} = \dx{x_k}$, by \cref{eq:shortest}, it cannot be ever rejected along the process $\op{x}$. 
	This implies $k < k'$, and hence, there is some $m' \in \{k, \dots, k'-1\}$ such that $x_k \in \RH{x}{m'}$. 
	If so, however, $\op{x}' := (x_1, \dots, x_{m'})$ should satisfy both $\set (\op{x}') \subseteq \set (\op{y})$ and $x_k \in \RHop{x'} - \RHop{y}$; this contradicts the shortest-length assumption on $\op{x}$.

Before we proceed, we introduce two more definitions:
Let $\ordsucc_D$ denote a preference profile such that $\op{y}$ is complete at it; by definition, $\cop (\ordsucc_D) = C_H (\set (\op{y}))$.
Also let $Z := \gs - \left[ C_H (\set(\op{x}))\cup C_H (\set(\op{y})) \right]$ be the set of contracts that are not chosen either from $\set(\op{x})$ or from $\set(\op{y})$.
Notice that $\drp{\dx{x_k}}{Z}$ has only one acceptable contract, $x_k$, because $\dx{x_k}$ signs no non-null contract at $\CHop{x}$ as seen above.
In what follows, we derive a contradiction for two cases separately. 
More specifically, we divide the case into two depending on the existence of an offer process $\op{w} = (w_1, \dots, w_T)$ such that 
	\begin{itemize}
	\item 	it is complete at $\drp{D}{Z}$, and 
	\item 	$\{w_1, \dots, w_t \} = \CHop{x}$ and $w_{t+1} = x_k$, where $t := \left| \CHop{x} \right|$. 
	\end{itemize}

\header{Case 1: the existence of $\op{w}$}
If such a process exists, $w_{t+1} = x_k$ must be an element of $\RH{w}{t+1} - \RHop{w}$ for the following reasons:
On the one hand, $w_{t+1} = x_k$ cannot be chosen from $\left\{w_1, \dots, w_{t+1}\right\}$.
To see this, note that the IRC implies $C_H (M') = C_H (M)$ for any menu $M$ and $M'$ such that $C_H (M) \subseteq M' \subseteq M$.
Applying this with $M = \set (\op{x})$ and $M' = \{w_1, \dots, w_{t+1} \}$, 
	we obtain $C_H (\{w_1, \dots, w_{t+1}\}) = C_H ( \set (\op{x})) = \{w_1, \dots, w_t\}$, where the second equality follows from the definition of $\op{w}$. 
On the other hand, $w_{t+1} $ must be chosen from $\set(\op{w})$.
To see this, note that each agent should weakly prefer $ \cop \left( \drp{D}{Z} \right) = \CHop{w}$ to $\cop \left( \ordsucc_D \right) = \CHop{y}$, by the assumption of \mon.
Since $w_{t+1} = x_k $ belongs to $\CHop{y}$ and it is the only acceptable contract for $\drp{\dx{w_{t+1}}}{Z}$, it should follow that $w_{t+1} \in \CHop{w}$.

However, $w_{t+1} \in \RH{w}{t+1} - \RHop{w}$ contradicts the assumption of OSaA: 
	At the step when $w_{t+1}$ is ``chosen back'' along the process $\op{w}$, OSaA requires that agent $\dx{w_{t+1}}$ should hold another non-null contract.
	This is impossible because $\drp{\dx{w_{t+1}}}{Z}$ has no other acceptable contract.

\header{Case 2: the nonexistence of $\op{w}$}

Now suppose that no offer process $\op{w}$ meets the conditions specified above.
Note that we can easily construct $\op{w}$ if $x_\ell \drpeq{\dx{x_\ell}}{Z} C_H (\set(\op{y}))$ for any $\ell \neq k$ such that $x_\ell \in \CHop{x}$.
The non-existence thus implies that 
		\begin{align} \label{eq:asm:ranking}
		\Bigl[ C_H (\set(\op{y})) \succ_{\dx{x_\ell}} x_\ell \mbox{ and } x_\ell \in \CHop{x} \Bigr]\mbox{ for some }\ell \leq m.
		\end{align}
In what follows, we fix any such $\ell$ and let $y^*$ denote the non-null contract agent $\dx{x_\ell}$ signs at $C_H (\set(\op{y}))$.
The rest of this proof is in two steps:
To begin, we derive a contradiction assuming $\op{x}$ is a subprocess of $\op{y}$. 
This means that a contradiction occurs if $C_H$ violates OS (rather than strong OS).
Thus, we then consider the general case assuming $C_H$ is OS.

To begin, suppose for the moment that $\op{x}$ is a subprocess of $\op{y}$.
Along the process $\op{y}$, then, $y^*$ must be offered before $x_\ell$ because $\op{y}$ is complete at $\ordsucc_D$ by definition and $\dx{x_\ell}$ prefers $y^*$ to $x_\ell$ by \cref{eq:asm:ranking}.
Under the subprocess assumption, this means that $y^* = x_{\ell'}$ for some $\ell' < \ell \leq m$.
Moreover, $y^* = x_{\ell'}$ should be rejected from $\{x_1, \dots, x_{\ell-1}\}$ for $\dx{x_\ell}$ to offer $x_\ell$ at step $\ell$, 
That is to say, $y^* = x_{\ell'}$ should be rejected from $\{x_1, \dots, x_\ell\}$ while it is not from $\set (\op{y})$ by the definition of $y^*$. 
This, however, contradicts \cref{eq:shortest}.

Let us now return to the general case where $\op{x}$ may not be a subprocess of $\op{y}$.
It is by now without loss of generality to assume that $C_H$ is OS. 
Then, $y^* \succ_{\dx{x_\ell}} x_{\ell}$ and $y^* \in \CHop{y}$ jointly imply that $x_\ell$ is \emph{not} offered along the process $\op{y}$ because under OS, $y^* \in \CHop{y}$ requires that $y^*$ be never rejected along $\op{y}$. 
However, since $x_\ell \in \set (\op{x})$ by definition, this contradicts the assumption of $\set (\op{x}) \subseteq \set (\op{y})$. 
\qed


\subsubsection{Proof of (iv) $\bm\Rightarrow$ (i)}

The proof is two-fold: We first show that strong OS implies weak Maskin monotonicity and then strengthen it to IR monotonicity.

\header{Part 1: weak Maskin Monotonicity}
Suppose that $C_H$ is strongly OS and arbitrarily fix $\ordsucc_D, \ordvarsucc_D \in \pdom$
	such that $\succ_{d}$ is a monotonic transformation of $\varsucc_{d}$ at $\xd{d}{\cop (\ordvarsucc_D)}$ for each $d \in D$. 
In what follows, we establish that $\cop (\ordsucc_D) \succeq_{d} \cop (\ordvarsucc_D)$ for all $d \in D$.

To begin with, let us consider a special case where  $\Ac (\ordsucc_{d}) \subseteq \Ac (\ordvarsucc_{d})$ for all $d \in D$.
Towards a contradiction, suppose further that $\cop (\ordvarsucc_D) \succ_{d^\circ} \cop (\ordsucc_D)$ for some $d^\circ \in D$.
Since $\cop (\succ_D) \succeq_{d^\circ} \varnothing$ by the individual rationality of the COM, 
	this implies $\cop (\varsucc_D) \succ_{d^\circ	} \varnothing$, and hence, 
	$d^\circ$ should sign some non-null contract $y^\circ$ at $\cop (\varsucc_D)$.
Then, it should be also offered but rejected along the COP with $\ordsucc_D$; 
		that is, $\Delta_R := R_H (\set(\op{x})) - R_H (\set(\op{y}))$ contains $ y^\circ$ and is thus non-empty, where $\op{x}$ and $\op{y}$ are complete offer processes at $(\ordsucc_D, C_H)$ and $(\ordvarsucc_D, C_H)$, respectively.
Applying \Cref{lem:sOS}, there should exist $x^* \in \set( \op{x})$ such that 
		(i) $x^* \not\in  \Ac \left( \ordvarsucc_{\dx{x^*}} \right)$ or 
		(ii) $x^* \succ_{\dx{x^*}} y^*$ and $y^* \varsucc_{\dx{x^*}} x^*$, 
		where $y^* = \xd{\dx{x^*}}{\cop (\ordvarsucc_D)}$.
The first case is impossible under the current assumption of $\Ac (\ordsucc_{\dx{x^*}}) \subseteq \Ac (\ordvarsucc_{\dx{x^*}})$.
The second case is also impossible, by the assumption that 
	$\succ_{\dx{x^*}}$ is a monotonic transformation of $\varsucc_{\dx{x^*}}$ at $ \xd{\dx{x^*}}{\cop (\ordvarsucc_D)}=y^* $.
To avoid a contradiction, 
	therefore, we must have $\cop (\ordsucc_D) \succeq_d \cop (\ordvarsucc_D)$ for all $d \in D$, 
	as long as  $\Ac (\ordsucc_{d}) \subseteq \Ac (\ordvarsucc_{d})$ for all $d \in D$.

To complete this part of the proof, now consider the general case where $\Ac (\ordsucc_{d}) \subseteq \Ac (\ordvarsucc_{d})$ may fail to hold.
Let $Z := \{ z \in \gs:  \cop (\ordvarsucc_D) \succ_{d} z \succ_{d} \varnothing \mbox{ for some } d \in D \}$.
For each $d \in D$, then, $\orddrp{d}{Z}$ remains a monotonic transformation of $\ordvarsucc_{d}$ at $\xd{d}{\cop (\ordvarsucc_D)}$, 
	while we regain $\Ac \left(\orddrp{d}{Z}\right) \subseteq \Ac (\ordvarsucc_{d})$.
Therefore, the conclusion of the previous paragraph entails that $\cop \left(\orddrp{D}{Z}\right) \succeq_{d} \cop (\ordvarsucc_D)$ for all $d \in D$.
This further implies that a complete offer process at $\left(\orddrp{D}{Z}, C_H \right)$ is also complete at $(\ordsucc_D, C_H)$.
Thus, we should have $\cop (\ordsucc_D) = \cop \left(\orddrp{D}{Z}\right) \succeq_{d}\cop (\ordvarsucc_D)$ for all $d \in D$, as desired.

\header{Part 2: IR Monotonicity}

Now we proceed to the proof of IR monotonicity.
Continue assuming $C_H$ is strongly OS.
Let $\ordsucc_D$ and $\ordvarsucc_D$ be an arbitrary pair of preference profiles such that 
	each $\ordsucc_d$ is an IR monotonic transformation of $\ordvarsucc_d$ at $\xd{d}{\cop(\ordvarsucc_D)}$.
What we need to show is $\cop \left( \ordsucc_D \right) \succeq_d \cop \left(\ordvarsucc_D\right)$ for all $d \in D$.

To begin, for each $d \in D$, define an ``intermediate'' preference $\ordsucc^0_d$ so that 
	(i) $\ordsucc^0_d$ is monotonic transformation of $\ordvarsucc_d$ at $\xd{d}{\cop(\ordvarsucc_D)}$, and 
	(ii) either $\ordsucc_d = \ordsucc^0_d$ or $\ordsucc_d = \left( \ordsucc^0_d \right)^{-x_d}$, where
			$x_d =\xd{d}{\cop(\ordvarsucc_D)} \neq \varnothing$.
More specifically, we can obtain such $\ordsucc^0_d$ for each $d \in D$ as follows:
\begin{itemize}
\item 	If $\ordsucc_d$ is a monotonic transformation of $\ordvarsucc_d$ at $\xd{d}{\cop(\ordvarsucc_D)}$,
				then $\ordsucc^0_d := \ordsucc_d$.
\item 	Otherwise, define $\ordsucc^0_d$ by 
				$\Ac (\ordsucc^0_d) := \left\{ w \in \gs: w \succeq_d x_d \right\}$ and 
				$w \succ^0_d w' \stackrel{\mathrm{def}}{\Longleftrightarrow} w \succ_d w'$ for all $w,w' \in \gs_d$.
\end{itemize} 
Note that the second case (where $\ordsucc^0_d \neq \ordsucc_d$) arises only if $d$ signs some non-null contract $x_d$ at $\cop (\ordvarsucc_D)$,
since an IR monotonic transformation at $\varnothing$ is always a monotonic transformation at $\varnothing$ by definitions.
When $\ordsucc^0_d \neq \ordsucc_d$, further, 
	$x_d$ is the least preferred acceptable contract for $\ordsucc^0_d$, and hence, $\ordsucc_d = \left( \ordsucc^0_d \right)^{-x_d}$;
	this is because if $x_d \in \Ac (\ordsucc_d)$, then the IR monotonic transformation $\ordsucc_d$ should be a monotonic transformation.
Since $\ordsucc^0_D$ is a monotonic transformation of $\ordvarsucc_D$ at $\cop (\ordvarsucc_D)$, 
it follows from weak Maskin monotonicity we have already established above,  
\begin{align} \label{eq:thm:IRMT:1}
\cop \left( \ordsucc^0_D\right) \succeq^0_d \cop (\ordvarsucc_D) \mbox{ for all } d \in D.
\end{align}

Now, arbitrarily label the agents as $\{d_1, \dots, d_{T}\} =D$, where $T:= |D|$, 
	and construct a sequence of preference profiles,
	$\ordsucc_D^{1}\dots, \ordsucc_D^{T}$, such that for each $i,t  \in \{1,\dots, T\}$,
	we have $\ordsucc_{d_i}^t := \ordsucc_{d_i}$  if $t \geq i$ and $\ordsucc_{d_i}^t := \ordsucc^0_{d_i}$ otherwise.
For each $t \in \{1, \dots, T\}$, we can then show that 
	\begin{align} \label{eq:thm:IRMT:2}
		\cop \left( \ordsucc^t_D\right) \succeq^t_d \cop (\ordsucc^{t-1}_D) \mbox{ for all } d \in D,
	\end{align} as follows.
Note that for each $d \in D$, by construction, either $\ordsucc^t_d = \ordsucc^0_d$ or $\ordsucc^t_d = \ordsucc_d = \left( \ordsucc^0_d \right)^{-x_d}$, where $x_d$ is the non-null contract $d$ signs at $\cop(\ordvarsucc_D)$. 
As a consequence, either $ \ordsucc^t_D = \ordsucc^{t-1}_D$ or $ \ordsucc^t_D = \left( \ordsucc^{t-1}_D \right)^{-x_{d_t}}$ holds by construction for each $t \in \{1,\dots, T\}$. 
If $ \ordsucc^t_D = \ordsucc^{t-1}_D$ and hence $\cop \left( \ordsucc^t_D \right) = \cop \left( \ordsucc^{t-1}_D \right)$, then \cref{eq:thm:IRMT:2} is trivial.
Otherwise, for any $d \neq d_t$, 
	$\cop \left( \ordsucc^t_D \right) \succeq_{d}^t \cop \left( \ordsucc^{t-1}_D \right)$ 
	follows from \Cref{lem:weak:COP} \ref{sublem:weak:COP:least}, since $ \ordsucc^t_D = \left( \ordsucc^{t-1}_D \right)^{-x_{d_t}}$ and $x_{d_t}$ is the least acceptable contract for $\ordsucc^{t-1}_{d_t}$.
Moreover, $\cop \left( \ordsucc^t_D \right) \succeq_{d_t}^t \cop \left( \ordsucc^{t-1}_D \right)$ also holds for the following reasons:
	\begin{itemize}
	\item Agent $d_t$ should sign some non-null contract at $\cop \left( \ordsucc^{t-1}_D \right)$,
						since
					\begin{align*}
					\cop \left( \ordsucc^{t-1}_D \right) \succeq^0_{d_t} \cop \left( \ordsucc^{t-2}_D \right) \succeq^0_{d_t} \cdots \succeq^0_{d_t} \cop \left( \ordsucc^{0}_D \right) \succeq^0_{d_t} \cop(\ordvarsucc_D) \ni x_{d_t},
					\end{align*}
					where for each $\tau < t$,  the ranking between $\tau$ and $\tau-1$ holds either by $ \ordsucc^\tau_D = \ordsucc^{\tau-1}_D$ or by \Cref{lem:weak:COP} \ref{sublem:weak:COP:least}, as we have just argued above.
	\item If $d_t$ signs $x_{d_t}$ at $ \cop \left( \ordsucc^{t-1}_D \right)$,
							then she should be assigned the null contract at $\cop \left( \ordsucc^{t}_D \right)$.
					This is because $ \cop \left( \ordsucc^{t-1}_D \right) \succ_{d_t}^{t-1}  \cop \left( \ordsucc^{t}_D \right)$ 
							by \Cref{lem:weak:COP} \ref{sublem:weak:COP:chosen} with $\ordsucc^t_D = \left( \ordsucc^{t-1}_D \right)^{-x_{d_t}}$, 
							and	$x_{d_t} \in \cop \left( \ordsucc^{t-1}_D \right)$ is the least-preferred acceptable contract for $\ordsucc^{t-1}_{d_t}$.
					Nevertheless, this implies $\cop \left( \ordsucc^{t}_D \right) \succ^t_{d_t} \cop \left( \ordsucc^{t-1}_D \right) $,
							since $x_{d_t}$ is unacceptable for $\ordsucc^t_{d_t} = \left( \ordsucc^{t-1}_{d_t}\right)^{-x_{d_t}}$.
	\item If she signs another contract $y_{d_t}$ at $ \cop \left( \ordsucc^{t-1}_D \right)$,
							then $y_{d_t} \succ^{t-1}_{d_t} x_{d_t}$ because $x_{d_t}$ is the least preferred acceptable.
					Hence, \Cref{lem:weak:COP} \ref{sublem:weak:COP:unpref} implies 
							$\cop \left( \ordsucc^t_D \right) = \cop \left( \ordsucc^{t-1}_D \right)$.								
	\end{itemize} 
Therefore, we should have \cref{eq:thm:IRMT:2} for all $d \in D$ and all $t \in \{1, \dots, T\}$.
		
		\medskip
		
Now we are ready to establish $\cop \left( \ordsucc_D \right) \succeq_d \cop \left(\ordvarsucc_D\right)$ for all $d \in D$.
Combining  \cref{eq:thm:IRMT:1,eq:thm:IRMT:2} across $t$'s, we obtain
\begin{align*}
	\cop (\ordsucc_D) \equiv \cop \left( \ordsucc_D^{T} \right) \succeq_d^T
		 \cop \left( \ordsucc_D^{T-1} \right) \succeq_{d}^{T-1} \cdots \succeq_{d}^1 
		 \cop \left( \ordsucc_D^{0} \right) \succeq_{d}^0 \cop ( \ordvarsucc_D).
\end{align*}
By the definitions of $\ordsucc_D^1, \dots, \ordsucc_D^T$, this particularly implies that for each $\tau \in \{1,\dots, T\}$,
\begin{align} \label{eq:thm:IRMT:3}
	\cop (\ordsucc_D) \succeq_{d_\tau}  \cop \left( \ordsucc_D^{\tau-1} \right) \succeq_{d_\tau}^0  \cop ( \ordvarsucc_D).
\end{align}
For $\cop (\ordsucc_D) \succeq_{d_\tau} \cop  ( \ordvarsucc_D)$ to fail to hold, thus, 
	$\ordsucc_{d_\tau}$ and $\ordsucc^0_{d_\tau}$ must disagree 
	on the ranking between $\cop \left( \ordsucc_D^{\tau-1} \right)$ and $\cop ( \ordvarsucc_D)$.
However, recall we defined $\ordsucc_{d_\tau}^0$ so that $\ordsucc_{d_\tau} \neq \ordsucc_{d_\tau}^0$ arises only if  
	there exists $x_{d_\tau} \in \cop ( \ordvarsucc_D)$ such that $\succeq_{d_\tau} = \left( \succeq_{d_\tau}^0 \right)^{-x_{d_\tau}}$.
That is to say, $\ordsucc_{d_\tau} \neq \ordsucc_{d_\tau}^0$ and $\cop \left( \ordsucc_D^{\tau-1} \right) \succeq_{d_\tau}^0  \cop ( \ordvarsucc_D)$ 
	jointly imply $x_{d_\tau} \in \cop (\ordvarsucc_D) - \Ac (\ordsucc_{d_\tau})$; if so, $\cop (\ordsucc_D) \succeq_{d_\tau}  \cop (\ordvarsucc_D)$ follows from the individual rationality of the COM. 
We can thus conclude from (\ref{eq:thm:IRMT:3}) that 
	$\cop (\ordsucc_D) \succeq_{d_\tau}  \cop (\ordvarsucc_D)$ for each $d_\tau \in D$, and the proof is complete.
\qed

 \subsection{Proof of \Cref{thm:uniqueness}} \label{subsec:proof:uniquness}

Since the ``if'' part is immediate by \Cref{thm:mon}, we only establish the ``only if'' part.
Let $f$ be a stable mechanism. 
We will show that it is not weakly Maskin monotonic if it is not the COM.
To begin, assume there exists $\ordsucc_D \in \pdom$ such that $f(\succ_D) \neq \cop(\succ_D)$.
Further assume that $\succ_D$ is ``minimal'' in the following sense: For any $\varsucc_D \in \pdom$,
	\begin{align} \label{eq:uni:1}
		\left[f(\varsucc_D) \neq \cop(\varsucc_D) \right]
		\Longrightarrow 
		\left[ \sum_{d\in D} \left| \Ac (\varsucc_d)\right| \geq \sum_{d\in D} \left| \Ac (\succ_d)\right| \right].
	\end{align}
This is without loss of generality because $\pdom$ is finite. 
Since both $f(\succ_D)$ and $\cop(\succ_D)$ are stable, by \citet[Lemma 1]{refall:1324}, 
	there is $d^* \in D$ who signs two distinct non-null contracts at $f(\succ_D)$ and $\cop(\succ_D)$.
To simplify notation, define $x:= \xd{d^*}{f(\succ_D)}$ and $x^\star := \xd{d^*}{\cop(\succ_D)}$.
The minimality assumption (i.e., \cref{eq:uni:1}) requires 
	\begin{align}\label{eq:uni:2} 
		f\left( \drp{D}{x} \right) = \cop \left( \drp{D}{x} \right)
		\mbox{ and }
		f\left( \drp{D}{x^\star} \right) = \cop \left( \drp{D}{x^\star} \right),
	\end{align}
	since $x, x^\star \in \Ac(\succ_{d^*})$ by the stability of $f$ and $\cop$.\footnote{
		Strictly speaking, $f\left( \drp{D}{x} \right)$ and $f\left( \drp{D}{x^\star} \right)$ are not well-defined because 
			$\drp{D}{x}$ and $\drp{D}{x^\star}$ are classes of preference profiles, 
			and $f$ may assign different outcomes for different profiles in each class. 
		However, the minimality assumption implies that \cref{eq:uni:2} holds for every element in $\drp{D}{x}$ and $\drp{D}{x^\star}$.
		Thus, $f\left( \drp{D}{x} \right)$ and $f\left( \drp{D}{x^\star} \right)$ are uniquely pinned down under the minimality assumption.
		}
In what follows, we divide the case into two and identify a violation of weak Maskin monotonicity for each case.

First, suppose $x \succ_{d^*} x^\star$. 
Then, we must have $ x^\star \succ_{d^*} \cop \left( \drp{D}{x^\star}\right)$ for the following reasons:
Let $(w_1, \dots, w_T)$ be a complete offer process that is observed during the COP algorithm at $\ordsucc_D$.
Since $x^\star$ is chosen as the final outcome, $w_{t^\star+1} = x^\star$ for some $t^\star$.\footnote{%
	Note $w_1 \neq x^\star$ because $d^*$ should offer $x$ before $x^\star$.
	}
Note also that any contract that $\succeq_{d^*}$ prefers to $x^\star$ should be rejected at step $t^\star$; i.e.,
	\begin{align} \label{eq:uni:3}
		\left[ y \succ_{d^*} x^\star \right] \Rightarrow 
		\left[ y = w_t \mbox{ for some } t \leq t^\star \mbox{ and } y \in R_H \left(\left\{w_1,\dots, w_{t^\star}\right\}\right)\right]. 
	\end{align}
Now consider the COP with $\orddrp{D}{x^\star}$.
We can run the algorithm so that $(w_1, \dots, w_{t^\star})$ are the first $t^\star$ offers. 
By the strong OS assumption, any contract $y \in R_H\left(\left\{w_1, \dots, w_{t^\star}\right\}\right)$ cannot be chosen at the final outcome.\footnote{%
	\label{ftn:asm:1}
	Note that this is valid even if we only assume OS.  
	}
Combining this observation with \cref{eq:uni:3}, we obtain $ x^\star \succ_{d^*} \cop \left( \drp{D}{x^\star}\right)$.
By \cref{eq:uni:2} and the supposition of $x \succ_{d^*} x^\star$, it then follows that $x \drp{d^*}{x^\star}  f\left( \drp{D}{x^\star}\right)$.
This means $f$ is not weakly Maskin monotonic because $\orddrp{d^*}{x^\star}$ is a monotonic transformation of $\ordsucc_{d^*}$ at $x = \xd{d^*}{f(\ordsucc_D)}$.

Second, suppose  $x^\star \succ_{d^*} x$. 
Under the strong OS assumption, then, $x$ is never offered during the COP with $\ordsucc_D$.\footnote{%
	\label{ftn:asm:2}
	Again, this is valid even if we only assume OS.  
	}
We thus have $\cop \left( \drp{D}{x}\right) = \cop \left(\ordsucc_D\right)$.
Then, by \cref{eq:uni:2}, $d^*$ should also be assigned $x^\star$ at $f \left( \drp{D}{x}\right)$.
By the supposition of $x^\star \succ_{d^*} x$, we obtain
	\begin{align*}
	 \xd{d^*}{f\left(\drp{D}{x}\right)} = x^\star \succ_{d^*} x = \xd{d^*}{f\left(\succ_D\right)}.
	\end{align*}
This means $f$ is not weakly Maskin monotonic because $\ordsucc_{d^*}$ is a monotonic transformation of $\orddrp{d^*}{x}$ at $x^\star = \xd{d^*}{f(\orddrp{D}{x})}$.
\qed

\subsection{Proof of \Cref{thm:GSP}} \label{subsec:proof:GSP}

As the ``only if'' part is immediate by definition, we only establish the ``if'' part.
Suppose towards a contradiction that $\cop (\cdot)$ is strategy-proof and that 
	there are $ \ordsucc^\circ_D, \ordvarsucc_D \in \pdom$ and $E \subseteq D$ such that 
	$\cop \left( \ordvarsucc_D \right) \succ^\circ_d  \cop\left( \ordsucc^\circ_D \right) $ for all $d \in E$ and 
	$\ordvarsucc_{d'} = \ordsucc^\circ_{d'}$ for all $d' \in D - E$.
Also assume $ \Ac (\ordsucc^\circ_d) \subseteq \Ac (\ordvarsucc_d) $ for all $d \in E$.
This is without loss of generality for the following reasons:
Suppose $w \in \Ac (\ordsucc^\circ_d) - \Ac (\ordvarsucc_d)$ for some $d \in E$.
	Let $\ordvarsucc'_d$  be the preference obtained by moving up $w$ to the ``bottom'' of acceptable contracts; that is,
		\begin{itemize}
		\item  	$z \varsucc'_d z' \Leftrightarrow z \varsucc_d z'$ for all $z, z' \in \gs_d - \{w\}$, 
		\item  	$z \varsucc'_d w $ for all $z \in \Ac (\ordvarsucc_d)$, and $w \varsucc'_d \varnothing$.
		\end{itemize}
Note that $d$ should sign a non-null contract at $\cop (\ordvarsucc_D)$ by the assumption of $\cop (\ordvarsucc_D) \succ^\circ_d \cop \left(\ordsucc^\circ_D \right)$.
During the COP, thus, $w$ is never offered no matter if it is acceptable or not; i.e., $\cop \left( \ordvarsucc'_d, \ordvarsucc_{-d}\right) = \cop (\ordvarsucc_D)$.
Repeating the same argument, we can construct $\ordvarsucc'_d, \ordvarsucc''_d, \dots, \ordvarsucc^{(n)}_d$
			so that $\cop (\ordvarsucc_D) = \cop \left(\ordvarsucc'_d, \ordvarsucc_{-d}\right) = \dots = \cop \left(\ordvarsucc^{(n)}_d, \ordvarsucc_{-d}\right)$ and $\Ac \left( \ordsucc^\circ_d\right) \subseteq \Ac \left(\ordvarsucc^{(n)}_D\right) $.
By redefining $\ordvarsucc_D$ to be $\ordvarsucc^{(n)}$, we can always guarantee $ \Ac (\ordsucc^\circ_d) \subseteq \Ac (\ordvarsucc_d) $ without changing the outcome of the COM.

Given $ \Ac (\ordsucc^\circ_d) \subseteq \Ac (\ordvarsucc_d) $ for all $d \in E$, construct another preference profile $\ordsucc_D$ from $ \ordsucc^\circ_D$ as follows:
	For each $d \in E$ and for each $w$ such that $\dx{w} = d$ and $w \succ^\circ_d \cop(\ordvarsucc_D) \varsucc_d w$,
	lower the ``position'' of $w$ to anywhere between the (possibly null) contracts that $d$ signs at $\cop(\ordvarsucc_D)$ and at $\cop(\ordsucc^\circ_D)$.
More formally, $\ordsucc_D$ is a preference profile such that for all $d \in E$, 
	\begin{itemize}
	\item $\Ac (\ordsucc_d) = \Ac \left(\ordsucc_d^\circ\right)$ for all $d \in E$, 
	\item $\left\{ z: z \succeq_d \cop \left( \ordsucc^\circ_D \right) \right\} = \left\{ z : z \succeq^\circ_d \cop \left(\ordsucc^\circ_D \right) \right\}$ for all $d \in E$,
	\item $\left\{ z: z \succ_d \cop \left(\ordvarsucc_D\right) \right\} \subseteq \{ z : z \varsucc_d \cop(\ordvarsucc_D) \} $ for all $d \in E$, and 
	\item $\ordsucc_{d'} = \ordsucc^\circ_{d'} = \ordvarsucc_{d'}$ for all $d' \in D -E $.
	\end{itemize} 
By construction,  for any $d \in D$, the ranking between $\cop (\ordvarsucc_D)$ and $\cop(\ordsucc_D)$ remains unchanged either with $\ordsucc^\circ_d$ or with $\ordsucc_d$.
Moreover, we also have 	$\cop (\ordsucc_D) = \cop (\ordsucc^\circ_D)$ by repeatedly applying \Cref{lem:non-bossy:OS+SP}.\footnote{
	More precisely, we can establish this equality as follows:
	Arbitrarily order the members of $E$ as $d_1, \dots, d_n$, and for each $k \in \{1,\dots, n\}$,   
		let $\succ_D^{k}$ to be $ \ordsucc^{k}_d = \ordsucc_d$ for all $d \in \{d_1 ,\dots, d_k\}$ and $\ordsucc^{k}_{d'} = \ordsucc^\circ_{d'}$ for all the others.
	Then, \Cref{lem:non-bossy:OS+SP} implies $ \cop \left( \ordsucc^\circ_D \right) =  \cop \left( \ordsucc^1_D \right) = \dots = \cop \left( \ordsucc^n_D \right) \equiv \cop (\ordsucc_D)$.
	}
These observations together imply $\cop (\ordvarsucc_D) \succ_d \cop(\ordsucc_D)$ for each $d \in E$.

We are now ready to derive a contradiction.
For any $d \in E$, it follows from $\cop (\ordvarsucc_D) \succ_d \cop(\ordsucc_D)$
that she should sign a non-null contract at $\cop (\ordvarsucc_D)$ and this contract should be offered but rejected along the COP with $\ordsucc_D$.
That is, $R_H (\set(\op{x})) - R_H (\set(\op{y}))$ is non-empty, where $\op{x}$ and $\op{y}$ are the complete offer processes at $(\ordsucc_D, C_H)$ and $(\ordvarsucc_D, C_H)$, respectively.
Applying \Cref{lem:sOS},  there should exist $x^* \in \set (\op{x})$ such that either 
	[1] $x^* \not\in  \Ac \left( \ordvarsucc_{\dx{x^*}} \right)$ or 
	[2] $y^* \not\succ_{\dx{x^*}} x^*$ and $y^* \varsucc_{\dx{x^*}} x^*$, 
	where $y^*$ is the (non-null) contract $\dx{x^*}$ signs at $\cop (\ordvarsucc_D)$.
If $\dx{x^*} \not\in E$, either case clearly contradicts the construction that $\ordsucc_{d'} = \ordvarsucc_{d'}$ for all $d' \not\in E$.
Even if $\dx{x^*} \in E$, the first case is impossible because for all $d \in E$,  $\Ac (\ordsucc_d) = \Ac \left(\ordsucc_d^\circ\right)$ by the construction of $\ordsucc_d$ and  $\Ac \left(\ordsucc_d^\circ\right) \subseteq \Ac (\ordvarsucc_d)$ by assumption.
The second case also contradicts the construction of $\ordsucc_D$ since $\left\{ z: z \succ_d \cop \left(\ordvarsucc_D\right) \right\} \subseteq \{ z : z \varsucc_d \cop(\ordvarsucc_D) \} $ ensures $\cop (\ordvarsucc_D) \varsucc_d x^* \Rightarrow \cop (\ordvarsucc_D) \succ_d x^*$ for any $d \in E$.
\qed

\subsection{Proof of \Cref{cor:suff}} \label{subsec:proof:suff}
Under the assumptions, the COM is known to be strategy-proof \citep[Theorem 7]{refall:1401}.
The claim is thus an immediate corollary of \Cref{thm:GSP}.
\qed

\subsection{Proof of \Cref{thm:anti-nec}} \label{subsec:proof:anti-nec}

Define a single-institution market $\mathcal{M}^0 = \left( D^0, \{h\}, X^0, C_h^{0} \right)$ as follows: 
Let $D^0 := \{d_1, d_2, d_3\}$, $X^0 := \{x_1, x_2, y_2, x_3, y_3\}$, 
	where $\dx{x_1} = d_1$, $\dx{x_2}= \dx{y_2} = d_2$, and $\dx{x_3}= \dx{y_3} = d_3$. 
Note that $x_1$ is the only contract in $X^0$ that involves agent $d_1$.
Also let $C_h^{0}$ be a choice function induced by a preference relation $\succ_h$ over the subsets of $X^0$ such that 
	\begin{align*}
		\{x_1, y_2, y_3\} 	&\succ_h \{x_1, y_3\} \succ_h \{ y_2, x_3\} \succ_h \{x_1,x_2\} \\
								&\succ_h [\mbox{all the other feasible doubletons}] \succ_h [\mbox{all the singletons}] \succ_h \emptyset, 
	\end{align*}
	where any tripleton except $\{x_1,y_2, y_3\}$ is unacceptable.
As we have seen in \Cref{ex:not-sOS}, $C^0_h$ is OS but not strongly OS. 
Now let $\mathcal{M} = (D, H, \gs, C_H)$ be an arbitrary regular extension of $\mathcal{M}^0$.
In what follows, our goal is to establish weak group strategy-proofness for the COM $\cop$ in the regular extension $\mathcal{M}$. 

To do so, we define three markets that differ from $\mathcal{M}$ only in institution $h$'s choice function.
In the first market, $\mathcal{M}^{-x_2}$, institution $h$'s choice function is given by $C^{-x_2}_h$, where
	\begin{align*}
		C^{-x_2}_h (\menu) = C^0_h (\menu - \{x_2\}) \mbox{ for each } \menu \InMenus.
	\end{align*} 
To define the second and third, consider two preference relations, $\vvsucc_h$ and $\vvvsucc_h$, over $X^0$ such that 
	\begin{align*}
		\{x_1, y_2, y_3\} 	&\vvsucc_h \{x_1, y_3\} \vvsucc_h {\color[named]{NavyBlue} \pmb{\{x_1, y_2 \}}} \vvsucc_h \{y_2, x_3\} \vvsucc_h \{x_1,x_2\} \\
								&\vvsucc_h [\mbox{all the other feasible doubletons}] \vvsucc_h [\mbox{all the singletons}] \vvsucc_h \emptyset, 
	\end{align*}
	where (i) any tripleton except $\{x_1,y_2, y_3\}$ is unacceptable and (ii) the rankings among the ``other'' doubletons and singletons are the same as under $\succ_h$,
	and 
	\begin{align*}
		 {\color[named]{NavyBlue} \pmb{\{y_2, y_3\}}} 	&\vvvsucc_h \{x_1, y_3\} \vvvsucc_h \{ y_2, x_3\} \vvvsucc_h \{x_1,x_2\} \\
								&\vvvsucc_h [\mbox{all the other feasible doubletons}] \vvvsucc_h [\mbox{all the singletons}] \vvvsucc_h \emptyset, 
	\end{align*}
	where (i) any tripleton is unacceptable and (ii) the rankings among the ``other'' doubletons and singletons are the same as under $\succ_h$.
Let $C^{\vvsucc}_h$ and $C^{\vvvsucc}_h$ be choice functions induced by $\vvsucc_h$ and $\vvvsucc_h$, and 
	define $\mathcal{M}^{\vvsucc}$ and $\mathcal{M}^{\vvvsucc}$ to be the markets where $h$'s choice function is $C^{\vvsucc}_h$ and $C^{\vvvsucc}_h$ (keeping all the other institutions' choice functions constant). 
In the same way as we did in \Cref{ex:not-sOS}, 
	one can easily confirm with \Cref{table:anti-nec} below that $C^{-x_2}_h$, $C^{\vvsucc}_h$, and $C^{\vvvsucc}_h$ are all strongly OS as well as strongly OSM. 
As a consequence, by \Cref{cor:suff}, \emph{the COM is weakly group strategy-proof in any of $\mathcal{M}^{-x_2}$, $\mathcal{M}^{\vvsucc}$, and $\mathcal{M}^{\vvvsucc}$.}
For the rest of this proof, we use $\cop_{-x_2}$, $\cop_{\vvsucc}$, and $\cop_{\vvvsucc}$, respectively, to denote the COM in those three markets. 

\begin{table}[p!]
	\renewcommand{\arraystretch}{1.5}
	\begin{minipage}[c]{\textwidth}
			\begin{center}
					\begin{tabular}{ccccccc} \toprule[1.5pt]
					$\set \left( \op{w}^3 \right)$ 	&$R_h \left(\set \left( \op{w}^3 \right)\right)$  
					&$w_4$ &$R_h \left(\set \left( \op{w}^4 \right)\right)$ &$w_5$ 	&$R_h \left(\set \left( \op{w}^5 \right)\right)$  \\ \midrule
					$\{x_1, x_2, x_3\}$ 			&{\color[named]{NavyBlue}$\pmb{\{x_2\}}$} 									&$y_2$ 		&$\{x_1, x_2\}$		&	&\\
					$\{x_1, x_2, y_3\}$ 			&$\{x_2\}$  							&$y_2$		&$\{x_2\}$ \\
					$\{x_{1}, y_{2}, x_{3}\}$ 	&$\{x_1\}$ 	& \\ 
					$\{x_{1}, y_{2}, y_{3}\}$ 	&$\emptyset$ \\ 
					\bottomrule[1.5pt]
					\end{tabular} 
			\end{center}
		\subcaption{Observable offer processes for $C^{-x_2}_h$ }
		\vspace{1em}
			\begin{center}
					\begin{tabular}{ccccccc} \toprule[1.5pt]
					$\set \left( \op{w}^3 \right)$ 	&$R_h \left(\set \left( \op{w}^3 \right)\right)$  
					&$w_4$ &$R_h \left(\set \left( \op{w}^4 \right)\right)$ &$w_5$ 	&$R_h \left(\set \left( \op{w}^5 \right)\right)$  \\ \midrule
					$\{x_1, x_2, x_3\}$ 			&$\{x_3\}$ 												&$y_3$ 		&$\{x_2, x_3\}$		&$y_2$ 	&{${\{x_2, x_3\}}$}\\
					$\{x_1, x_2, y_3\}$ 			&$\{x_2\}$ 												&$y_2$		&$\{x_2\}$ \\
					$\{x_{1}, y_{2}, x_{3}\}$ 	&{\color[named]{NavyBlue}$\pmb{\{x_3\}}$} 	&$y_3$		&$\{x_3\}$	\\ 
					$\{x_{1}, y_{2}, y_{3}\}$ 	&$\emptyset$ \\ 
					\bottomrule[1.5pt]
					\end{tabular} 
			\end{center}
		\subcaption{Observable offer processes for $C_h^{\ordvvsucc}$ } 
		\vspace{1em}

			\begin{center}
					\begin{tabular}{ccccccc} \toprule[1.5pt]
					$\set \left( \op{w}^3 \right)$ 	&$R_h \left(\set \left( \op{w}^3 \right)\right)$  
					&$w_4$ &$R_h \left(\set \left( \op{w}^4 \right)\right)$ &$w_5$ 	&$R_h \left(\set \left( \op{w}^5 \right)\right)$  \\ \midrule
					$\{x_1, x_2, x_3\}$ 			&$\{x_3\}$ 												&$y_3$ 		&$\{x_2, x_3\}$		&$y_2$ 	&{\color[named]{NavyBlue}$\pmb{\{x_1, x_2, x_3\}}$}\\
					$\{x_1, x_2, y_3\}$ 			&$\{x_2\}$ 												&$y_2$		&{\color[named]{NavyBlue}$\pmb{\{x_1, x_2\}}$} \\
					$\{x_{1}, y_{2}, x_{3}\}$ 	&$\{x_1\}$  	\\ 
					$\{x_{1}, y_{2}, y_{3}\}$ 	&{\color[named]{NavyBlue}$\pmb{\{x_1\}}$}\\ 
					\bottomrule[1.5pt]
					\end{tabular} 
			\end{center}
		\subcaption{Observable offer processes for $C_h^{\vvvsucc}$ }
		\vspace{1em}
	\end{minipage} 
	\caption{Observable offer processes  in the proof of \Cref{thm:anti-nec}} \label{table:anti-nec}
\end{table}

Before we proceed, we make several observations on institution $h$'s choice functions.
In doing so, let $W$ denote the set of all observable offer processes in $\mathcal{M}$.
First, $C^0_h$ never rejects $y_2$, if offered, along any observable process; i.e.,
	\begin{align} \label{eq:almost:safe}
		\left[ \op{w} \in W \mbox{ and } y_2 \in \set (\op{w})  \right] \Rightarrow \left[ y_2 \in C^0_h (\set (\op{w})) \right].
	\end{align}
This can easily be confirmed with \Cref{table:not-sOS} above. 
Second, along any $\op{w} \in W$, $C^0_h$ and $C^{\vvsucc}_h$ disagree only when the menu is $\{ x_1, y_2, x_3\}$; i.e., 
	\begin{align} \label{eq:almost:diff1}
		\left[ \op{w} \in W \mbox{ and } C^0_h (\set (\op{w})) \neq C^{\vvsucc}_h (\set (\op{w})) \right] \Rightarrow 
		\left[\set (\op{w}) \cap X^0 = \{x_1, y_2, x_3\}\right].
	\end{align}
To see this, note that $C^0_h (M) \neq C^\vvsucc_h (M)$ only when $M\cap X^0 = \{x_1, y_2, x_3\}, \{x_1, x_2, y_2\}$, or $\{x_1,x_2, y_2, x_3\}$.
Among these possible menus, the second and third are unobservable as one can check with \Cref{table:not-sOS}.
Lastly, for any menu of contracts, $C^0_h$ and $C^{\vvvsucc}_h$ disagree only when the menu includes either $\{x_1, y_2, y_3\}$ or $\{y_2, x_3, y_3\}$. 
As a result, we have 
	\begin{align} \label{eq:almost:diff2}
		\left[ C^0_h (M) \neq C^{\vvvsucc}_h (M) \right] \Rightarrow \left[ \{y_2, y_3\} \subseteq M \right].
	\end{align}

To show that $\cop$ is weakly group strategy-proof in $\mathcal{M}$ by contradiction, suppose otherwise. 
That is, there are some $\ordsucc_D, \ordvarsucc_D \in \pdom$ such that 
	\begin{align} \label{eq:almost:violation}
		\cop \left( \ordvarsucc_D \right) \succ_d \cop \left( \ordsucc_D \right) \mbox{ for all } d \in E := \{ e\in D: \ordsucc_e \neq \ordvarsucc_e\} \neq \emptyset.
	\end{align}
Throughout the rest of the proof, let $\op{w}_\succ$ and $\op{w}_\varsucc$ be (arbitrary) complete offer processes at $\succ_D$ and $\varsucc_D$, respectively, in the market $\mathcal{M}$.  
We make two observations.
First,  $d_2$ cannot prefer $y_2$ to $x_2$ both at $\succ_{d_2}$ and at $\varsucc_{d_2}$, i.e.,
	\begin{align} \label{eq:almost:x2}
		[ y_2 \not\succ_{d_2} x_2] \mbox{ or }  [y_2 \ntriangleright_{d_2} x_2]
	\end{align}
for the following reasons:
Otherwise, neither $\op{w}_\succ$ nor $\op{w}_\varsucc$ contains $x_2$ because $C^0_h$ never rejects $y_2$ (\cref{eq:almost:safe}).
Then, the difference between $C^0_h$ and $C^{-x_2}_h$ is irrelevant along these processes, 
	and hence, we have both $\cop_{-x_2} ( \ordsucc_D) = \cop (\ordsucc_D)$ and  $\cop_{-x_2} ( \ordvarsucc_D) = \cop (\ordvarsucc_D)$. 
However, these equalities are incompatible with the assumption of \cref{eq:almost:violation} because $\cop_{-x_2}$ is weakly group strategy-proof as seen above.
Second, since $\cop_{\ordvvsucc}$ is also weakly group strategy-proof, 
	at least one of $\cop \left( \ordvarsucc_D \right) \neq \cop_{\ordvvsucc} \left( \ordvarsucc_D \right)$ and $\cop \left( \ordsucc_D \right) \neq \cop_{\ordvvsucc} \left( \ordsucc_D \right)$ must hold. 
For the rest of this proof, we establish a contradiction for each of these two cases.

\header{Case 1: ${\cop \left( \ordvarsucc_D \right) \neq \cop_{\ordvvsucc} \left( \ordvarsucc_D \right)}$.}

For $\cop \left( \ordvarsucc_D \right) \neq \cop_{\ordvvsucc} \left( \ordvarsucc_D \right)$,
	the two choice functions, $C_h^0$ and $C_h^\vvsucc$, must disagree on the choice from a subprocess, say $\widetilde{\op{w}}_\varsucc$, of $\op{w}_\varsucc$.
Then, $\set \left( \widetilde{\op{w}}_\varsucc \right) \cap X^0 = \{x_1, y_2, x_3\}$ by \cref{eq:almost:diff1},
	and moreover, institution $h$ should not receive any additional offers for the rest of the process.\footnote{
	Once $\{ x_1, y_2, x_3\}$ is reached, $d_1$ has no more contract to offer to $h$, and $d_2$ and $d_3$ hold $y_2$ and $x_3$.
	}
Thus, we obtain $\set (\op{w}_\varsucc) \cap X^0 = \{x_1, y_2, x_3\}$.
This entails $y_2 \varsucc_{d_2} x_2$, and hence, $y_2 \not\succ_{d_2} x_2$ should also hold by \cref{eq:almost:x2}.
These imply that $d_2$ is a deviator (i.e., a member of $E$ in \cref{eq:almost:violation}), and hence, $\cop \left( \ordvarsucc_D \right) \succ_{d_2} \cop \left( \ordsucc_D \right) $.
This, however, cannot be the case
	because $\cop \left( \ordsucc_D \right) \succeq_{d_2} y_2$ by $C^0_h$ never rejecting $y_2$, and 
	$y_2 \in C^0_h (\{x_1, y_2, x_3\}) \subseteq \cop (\varsucc_D)$.

\header{Case 2: ${\cop \left( \ordsucc_D \right) \neq \cop_{\ordvvsucc} \left( \ordsucc_D \right)}$.}

In this case, we obtain $\set (\op{w}_\succ) \cap X^0 = \{x_1, y_2, x_3\}$ as we did for $\op{w}_\varsucc$ in Case 1.
As a consequence, we also have $\{y_2, x_3\} \subseteq \cop \left( \ordsucc_D \right) $.
Two remarks are in order.
First, $y_3$ is not offered along $\op{w}_\succ$.
Second, $y_2$ cannot be offered along $\op{w}_\varsucc$ for the following reasons:
Since $x_2$ is not offered along $\op{w}_\succ$, we have $y_2 \succ_{d_2} x_2$.
By \cref{eq:almost:x2}, $x_2 \varsucc_{d_2} y_2$ is necessary. 
It follows that $d_2 \in E$ in \cref{eq:almost:violation}, and hence, $\cop \left( \ordvarsucc_D \right) \succ_{d_2} \cop \left( \ordsucc_D \right) $.
In light of $y_2 \in \cop \left( \ordsucc_D \right) $ and \cref{eq:almost:safe}, 
	this is possible only if $y_2$ is not offered along $\op{w}_\varsucc$.

Now we are ready to derive a contradiction. 
The above two observations imply that neither $\set (\op{w}_\succ)$ nor $\set (\op{w}_\varsucc)$ contains both $y_2$ and $y_3$. 
By \cref{eq:almost:diff2}, then, we should have $\cop \left( \ordsucc_D \right) = \cop_\vvvsucc \left( \ordsucc_D \right)$ and $\cop \left( \ordvarsucc_D \right) = \cop_\vvvsucc \left( \ordvarsucc_D \right)$.
Combined with \cref{eq:almost:violation}, this contradicts the weak group strategy-proofness of the COM in the market $\mathcal{M}^\vvvsucc$. 
\qed

\singlespacing
\bibliographystyle{ecta}
\bibliography{refall}
\onehalfspacing
\appendix

\crefalias{section}{appendix}
\newpage

		\section{More on the Characterizations of the COM} \label{sec:remarks}

In this appendix, we examine if we can replace weak Maskin monotonicity in \Cref{thm:uniqueness} with other related properties.
To begin, it is immediate that the theorem continues to hold with IR monotonicity (\Cref{cor:char:IR}).
In contrast, dropping monotonicity is insufficient to characterize the COM (\Cref{fact:char:dropping}), as we will show below.

\begin{cor} \label{cor:char:IR}
Let $C_H$ be a strongly OS profile of choice functions. 
Then, a mechanism is stable and IR monotonic if and only if it is the COM.
\end{cor}
\begin{proof}
The COM satisfies IR monotonicity by \Cref{thm:mon}. 
If a stable mechanism is IR monotonic, it must also be weakly Maskin monotonic and thus be the COM by \Cref{thm:uniqueness}.
\end{proof}

\begin{fact} \label{fact:char:dropping}
Let $C_H$ be a strongly OS profile of choice functions. 
Then, a non-COM mechanism can be stable and dropping monotonic.
\end{fact}
\begin{proof}
The proof is by example. 
Let $D=\{d\}$, $H=\{h\}$, and $\gs = \{x,y\}$, where $x,y$ are two distinct contracts between $d$ and $h$.
Define $C_h$ to be the choice function induced by the following preference:
\begin{align*}
    \{x\} \succ_h \{y\} \succ_h \emptyset.
\end{align*}
In this market, $\{x\}$ and $\{y\}$ are stable if and only if $x$ and $y$ are acceptable for $d$, respectively.
Note also that the COM always chooses $d$'s best acceptable contract, if any.

Now we present a non-COM mechanism that is stable and dropping monotonic.
Define $f$ as follows:
    \begin{align*}
        f(\ordvarsucc_d) = 
        \begin{cases}
            \{x\} &\mbox{ if } \ordvarsucc_d = \ordsucc_d \\
            \cop(\ordvarsucc_d) &\mbox{ otherwise, }
        \end{cases}
    \end{align*}
where $\succ_d$ is the preference such that $y \succ_d x \succ_d \varnothing$.
The stability of $f$ is obvious from the above observation.
To check dropping monotonicity, it suffices to compare $\ordsucc_d$ with $\orddrp{d}{x}$ and $\orddrp{d}{y}$.
Then, since
    \begin{align*}
    f\left(\orddrp{d}{x}\right) = \{y\} \drp{d}{x} \{x\}
    &= f(\ordsucc_d), \mbox{ and } \\
    f\left(\orddrp{d}{y}\right) = \{x\}  &= f(\ordsucc_d)  ,
    \end{align*}
$f$ is indeed dropping monotonic, and the proof is complete.
\end{proof}

Next we turn to the \emph{invariance to lower-tail preference changes} and \emph{truncation-proofness}.
Under unilateral substitutability, \citet{refall:1189} shows that each of these axioms also characterizes the COM within stable mechanisms.
For the invariance to lower-tail preference changes, the same remains true under strong OS (\Cref{cor:char:weakMaskin}). 
However, the characterization with truncation-proofness does not extend to strong OS (\Cref{fact:char:truncation}):
As we establish below, the COM is not necessarily the unique stable mechanism that is truncation-proof.

\begin{definition}
Given an agent $d$, preference $\ordsucc_d$, and contract $w \in \gs_d$, define $U(\ordsucc_d, w) := \{w' \in \gs: w' \varsucc_d w \}$.
A preference $\ordsucc_d$ of agent $d$ is a \emph{lower-tail transformation} of $\ordvarsucc_d$ at $w \in \gs_d$
	if $U(\ordsucc_d, w) = U(\ordvarsucc_d,w)$ and 
	$[x \succ_d y \Leftrightarrow x \varsucc_d y]$ for all $x,y \in U(\ordsucc_d, w)$. 
A mechanism $f$ is \emph{invariant to lower-tail preference changes} 
	if $f(\ordsucc_D) = f\left( \ordvarsucc_D\right)$ holds 
	whenever $\ordsucc_D, \ordvarsucc_D \in \pdom$ are such that for each $d \in D$, $\ordsucc_d$ is a lower-tail transformation of $\ordvarsucc_d$ at $\xd{d}{f(\ordvarsucc_D)}$.
\whiteqed
\end{definition}

\begin{definition}
A preference $\varsucc_d$ of agent $d$ is a \emph{truncation} of $\succ_d$ if 
	$\Ac(\ordvarsucc_d) \subseteq \Ac(\ordsucc_d)$ and 
	$[x \varsucc_d y \Leftrightarrow x \succ_d y]$ for all $x,y \in \gs_d$.
A mechanism $f$ is \emph{truncation-proof} 
	if $f(\ordsucc_D) \succeq_d f\left( \ordvarsucc_d, \ordsucc_{-d}\right)$
	holds for any $\succ_D \in \pdom$, $d \in D$, and truncation $\varsucc_d$ of $\succ_d$.
\whiteqed
\end{definition}

\begin{cor} \label{cor:char:weakMaskin}
Let $C_H$ be a strongly OS profile of choice functions. 
Then, a mechanism is stable and invariant to lower-tail preference changes if and only if it is the COM.
\end{cor}
\begin{proof}
This is a corollary of our \Cref{thm:uniqueness} because in its proof, each violation of weak Maskin monotonicity also violates the invariance to lower-tail preference changes.
\end{proof}

\begin{fact} \label{fact:char:truncation}
	Let $C_H$ be a strongly OS profile of choice functions. 
	Then, a non-COM mechanism can be stable and truncation-proof. 
\end{fact}
\begin{proof}
The proof is by example. 
Let $D = \{d_1, d_2, d_3\}$, $H=\{h\}$, $\gs = \{x_1, y_1, z_1, x_2, y_2, x_3\}$, where $\dx{x_1}=\dx{y_1}=\dx{z_1} = d_1$, $\dx{x_2}=\dx{y_2}=d_2$, and $\dx{x_3}=d_3$.
Suppose that $h$ has two slots and that its choice function $C_h$ is induced by \emph{slot-specific priorities} \citep{refall:1158}:
Let $\succ_{h1}$ and $\succ_{h2}$ be preference relations over contracts such that 
	\begin{align*}
		z_1 &\succ_{h1} x_3 \succ_{h1} x_2 \succ_{h1} x_1  \succ_{h1} \varnothing, \mbox{ and }\\
		y_2 &\succ_{h2} y_1 \succ_{h2} x_2 \succ_{h2} x_1  \succ_{h2}  \varnothing,
	\end{align*}
	where $\varnothing$ represents vacancy. 
From each menu of contracts, 
	the first slot chooses the best contract according to $\succ_{h1}$,
	and then, the second slot chooses the best contract based on $\succ_{h2}$ from the feasible remaining contracts. 
Note that the second slot cannot choose a contract involving the same agent as of the first slot's choice. 
When the menu is $\{y_1,z_1, x_2\}$, for instance,
	not only $z_1$ but also $y_1$ becomes infeasible after the first slot chooses $z_1$;
	thus, $C_h (\{y_1,z_1, x_2\}) = \{z_1,x_2\}$ in this case. 
Any choice function induced by slot-specific priorities satisfies strong OS.\footnote{%
	Slot specific priorities imply substitutable completability \citep{refall:1191}, which in turn implies strong OS \citep{refall:1401}.
	}
Hence, the COM is well-defined, stable, and IR monotonic in this market.

Next, we define a non-COM mechanism that we will show is also stable and truncation-proof. 
To do so, define $\vvsucc_D = (\ordvvsucc_{d_1}, \ordvvsucc_{d_2}, \ordvvsucc_{d_3})$ to be the preference profile such that 
	\begin{align*}
		x_1 \vvsucc_{d_1} y_1  \vvsucc_{d_1} z_1  &\vvsucc_{d_1} \varnothing,\\
		x_2 \vvsucc_{d_2} y_2  &\vvsucc_{d_2} \varnothing, \mbox{ and } \\ 
		x_3 &\vvsucc_{d_3} \varnothing.
	\end{align*}
Referring to this particular profile, define a mechanism $f$ as follows: For each $\ordsucc_D \in \pdom$, 
	\begin{align*}
		f (\ordsucc_D) = 
		\begin{cases}
			\{z_1, x_2\}		&\mbox{ if } \ordsucc_D = \ordvvsucc_D, \\
			\cop(\ordsucc_D) &\mbox{ otherwise. }
		\end{cases}
	\end{align*}
Notice that the COM outcome at $\ordvvsucc_D$ is $\{z_1, y_2\}$; hence, $f$ is not the COM. 
Further, $f (\ordvvsucc_D) = \{z_1, x_2\}$ is a stable allocation at $\ordvvsucc_D$; hence, $f$ is a stable mechanism.
In what follows, we confirm that $f$ is also truncation-proof.

Towards a contradiction, suppose that $f$ violates truncation-proofness; i.e., there exists $\ordsucc_D, \ordvarsucc_D \in \pdom$ such that 
	\begin{itemize}
		\item $f (\ordvarsucc_D) \succ_{d^*} f (\ordsucc_D)$ for some $d^* \in D$, whereas
		\item $\ordvarsucc_{d^*}$ is a truncation of $\ordsucc_{d^*}$, and $\ordvarsucc_d = \ordsucc_d$ for all $d \neq d^*$. 
	\end{itemize}
Then, either $\ordsucc_D$ or $\ordvarsucc_D$ must be $\ordvvsucc_D$ because, as mentioned above, the COM is IR monotonic in this market.
Obviously, $\ordvarsucc_D = \ordvvsucc_D$ is impossible because every contract is acceptable at $\ordvvsucc_D$ (and hence, it cannot be a truncation of any other). 
Thus, the only remaining possibility is $\ordsucc_D = \ordvvsucc_D$.
For $d_1$, it is impossible to benefit with a truncation 
	because given the preferences of the others, $d_1$ cannot obtain a non-null contract without reporting $z_1$ as acceptable.
For $d_2$, truncation cannot be profitable because she is assigned the first-best contract at $f(\ordvvsucc_D)$.
For $d_3$, the only truncation $\ordvarsucc_{d_3}$ is such that $\Ac (\ordvarsucc_{d_3})=\emptyset$, which is clearly unprofitable.
In summary, starting from $\ordvvsucc_D$, no agent can profitably manipulate $f$ via truncation; the mechanism $f$ is thus truncation-proof.
\end{proof}

		
		\section{Omitted Proofs} \label{sec:omitted_proofs}


In this appendix, we provide the proofs of \Crefrange{lem:sOS}{lem:non-bossy:OS+SP}.
All of those proofs are reproduced from \citet{refall:1401} for the reader's convenience only.
As such, they should not count as a part of the contribution of the present paper.

\subsection[Proof of \Cref{lem:sOS}]{Proof of \Cref{lem:sOS}} \label{sec:proof:lem:sOS}

	Suppose $\Delta_R := R_H (\set(\op{x})) - R_H (\set(\op{y}))$ is non-empty.
	Then, there exists the first step $n$ 
		at which any contract in $\Delta_R$ is rejected during the process 	$\op{x} = (x_1, \dots, x_n)$;
		that is, $n$ is such that $R_H (\{x_1, \dots x_{n-1}\}) \cap \Delta_R = \emptyset$ and $R_H (\{x_1, \dots, x_n\}) \cap \Delta_R \neq \emptyset$.
	The latter implies $R_H  (\{x_1, \dots, x_n\}) \not\subseteq R_H (\set ( \op{y}))$, 
			which further entails  $\{x_1, \dots, x_n\} \not\subseteq \set (\op{y})$ by the assumption of strong OS.
	That is, there exists $k \leq n$ such that $x_k \in \{x_1, \dots, x_n\} - \set (\op{y})$. 
		
	Now let $x^* = x_k$ and $d = \dx{x^*}$.
	If $x^* \in \Ac ( \ordvarsucc_d)$, then $d$ should sign some (non-null) contract $y^* \varsucc_d x^*$ at $\COP (\ordvarsucc_D, C_H)$;
			otherwise, $x^* \not\in \set (\op{y})$ contradicts the assumption that $\op{y}$ is complete at $( \ordvarsucc_D, C_H)$.
	Furthermore, $y^* \succ_d x^*$ is impossible for the following reasons:
			If $y^* \succ_d x^*$, then $y^*$ must be offered and rejected before $x^* = x_k$ is offered at step $k < n$ of the process $\op{x}$. 
	By the assumption of (strong) OS, it then follows that $y^* \in R_H (\{x_1, \dots, x_{n-1}\}) \subseteq R_H (\set (\op{x}))$, 
		which further entails $y^* \in R_H (\{x_1, \dots, x_{n-1}\}) \cap \Delta_R$ since $y^* \not\in R_H (\set (\op{y}))$ by its definition.
	This, however, contradicts the definition of $n$.
	As we have shown that $x^* \in \Ac ( \ordvarsucc_d)$ implies $y^* \varsucc_d x^*$ and $y^* \not\succ_d x^*$, the proof is complete.
\qed

\subsection{Proof of \Cref{lem:weak:COP}}
			To prove part \ref{sublem:weak:COP:unpref}, 
				suppose $\dx{w}$ signs a non-null contract $z$ at $\cop ( \ordvarsucc_D)$ and that $z \succ_{\dx{w}} w$.
			Let $\op{y} = (y_1, \dots, y_T)$ be a complete offer process at $ \left( \ordvarsucc_D, C_H \right)$.
			Then, $\set (\op{y})$ cannot contain $w$ for the following reasons:
				For $w$ to be offered, $z$ must be rejected beforehand. Under the assumption of OS, then, $z$ must be also rejected from $\set (\op{y})$, which contradicts its definition.
			Given $w \not\in \set (\op{y})$, it is immediate to see that $\op{y} = (y_1, \dots, y_T)$ is also complete at $ \left( \ordvardrp{D}{w}, C_H \right)$ and hence, 
				$\cop \left( \ordvardrp{D}{w} \right) = \cop ( \ordvarsucc_D)$.
				
			To prove part \ref{sublem:weak:COP:chosen}, 
				suppose $w \in \cop ( \ordvarsucc_D)$ and that $\op{y} = (y_1, \dots, y_T)$ is a complete offer process at $ \left( \ordvarsucc_D, C_H \right)$.
			Obviously, there exists some $t$ such that $y_t = w$. 
			By rerunning the cumulative offer process from step $t$ with $\ordvardrp{D}{w}$, 
				then, we can obtain an offer process $\op{y}' = \left( y_1, \dots, y_{t-1}, y'_t, \dots, y'_{T'} \right)$ that is complete at $ \left( \ordvardrp{D}{w}, C_H \right)$.
			By definitions, any contract better than $w$ for $\dx{w}$, with respect to either $\ordvarsucc_{\dx{w}}$ or $\ordvardrp{\dx{w}}{w}$, must be an element of and be rejected from $\{y_1, \dots y_{t-1}\}$.
			Under the assumption of OS, it must be also rejected from $\set (\op{y}')$.
			Therefore, $\cop ( \ordvarsucc_D) \varsucc_{\dx{w}} \cop \left( \ordvardrp{D}{w} \right)$.

			To prove part \ref{sublem:weak:COP:least}, 
				suppose $w' \varsucceq_{\dx{w}} w$ for all $w' \in \Ac \left( \ordvarsucc_{\dx{w}} \right)$. 
			If $\dx{w}$ signs a non-null contract $w'$ at $\cop \left( \ordvardrp{D}{w} \right)$, 
				then, $\cop \left( \ordvardrp{D}{w} \right) = \cop (\ordvarsucc_D)$ holds
				because any complete offer process at $\ordvardrp{D}{w}$ is also complete at $\ordvarsucc_D$ as in the proof of part \ref{sublem:weak:COP:unpref}.
			To complete the proof, 
					suppose that $\dx{w}$ signs no non-null contract at $\cop \left( \ordvardrp{D}{w} \right) = C_H \left(\set (\op{y}^-) \right)$.
			Let $\op{y}^- = \left( y^-_1, \dots, y^-_{T} \right)$ be a complete offer process at $(\ordvardrp{D}{w}, C_H)$.
			Then, we can restart the cumulative offer process from step $T+1$ by letting $\dx{w}$ offer $y''_{T+1} = w$, 
				so as to obtain an offer process \begin{align*}\op{y}'' = \left( y^-_1, \dots, y^-_T, y''_{T+1}, \dots, y''_{T''} \right)\end{align*} that is complete at $ \left( \ordvarsucc_D, C_H \right) $.
			By the OS assumption, any contract rejected from $\set (\op{y}^-)$ must be also rejected from $\set (\op{y}'')$.
			For any $d \neq \dx{w}$, thus, $\cop \left( \ordvardrp{D}{w} \right) \varsucceq_d \cop (\ordvarsucc_D)$.
\qed

\subsection{Proof of \Cref{lem:non-bossy:OS+SP}}
		Taking $d$ as arbitrarily fixed, suppose towards a contradiction that 
			$(\ordsucc_D, \ordvarsucc_D)$ is a counterexample; i.e., 
			the three conditions on $\ordsucc_D$ and $\ordvarsucc_D$ are satisfied while $\cop(\ordvarsucc_D) \neq \cop (\ordsucc_D)$.
		Without any loss of generality, suppose further that it is ``minimal'' in the following sense:
		For any other counterexample $(\ordsucc'_D, \ordvarsucc'_D)$, 
			\begin{align} \label{eq:lem:non-bossy}
			\min \left\{ \sum\limits_{d' \in D} \left| \Ac \left( \ordsucc'_{d'} \right) \right| , \sum\limits_{d' \in D} \left| \Ac\left (\ordvarsucc'_{d'} \right) \right| \right\}
			\ \geq \  
			\min \left\{ \sum\limits_{d' \in D} \Bigl| \Ac \left(\ordsucc_{d'} \right) \Bigr| , \sum\limits_{d' \in D}  \Bigl| \Ac \left(\ordvarsucc_{d'} \right) \Bigr| \right\}.
			\end{align}
		To complete the proof, then,	it suffices to construct a non-empty $Y$ such that 
			$(\ordsucc'_D, \ordvarsucc'_D) = \left( \orddrp{D}{Y}, \ordvardrp{D}{Y} \right)$ forms a counterexample violating this inequality.

		To begin with, note that $\cop (\ordvarsucc_D) =_d \cop (\ordsucc_D)$ should hold by the assumption of strategy-proofness:
		If $\cop (\ordvarsucc_D) \succ_d \cop(\ordsucc_D)$, then $d$ would have an incentive to report $\ordvarsucc_d$ when the true preference is $\ordsucc_d$.
		If $\cop (\ordsucc_D) \succ_d \cop (\ordvarsucc_D)$, then 
			the second condition in the statement of this lemma implies $\cop (\ordsucc_D) \varsucc_d \cop (\ordvarsucc_D)$.
		Thus, $d$ could benefit by reporting $\ordsucc_d$ when the true preference is $\ordvarsucc_d$.

		Next, we confirm that there should be some $d^* \in D$ who signs distinct non-null contracts at $\cop (\ordsucc_D)$ and $ \cop ( \ordvarsucc_D)$; i.e., 
			there should exist $x^*_\succ \in \cop (\ordsucc_D)$ and $x^*_\varsucc \in \cop(\ordvarsucc_D)$ such that $\dx{x^*_\succ} = \dx{x^*_\varsucc} = d^*$ and $x^*_\succ \neq x^*_\varsucc$.
		By Lemma 2 of \citet{refall:1324}, such $d^*$ is guaranteed to exist if $\cop (\ordsucc_D)$ is stable at $\left(\ordvarsucc_D, C_H \right)$.		
		For some $(h, X')$ to block $\cop(\ordsucc_D)$ at $(\ordvarsucc_D, C_H)$ but not at $(\ordsucc_D, C_H )$, 
			we must have $C_h (X') \varsucc_{d'} \cop (\ordsucc_D) \succ_{d'} C_h (X') $ for some $d' \in D$.
		However, this is clearly impossible under our assumptions;
			whether $d' = d$ or not,
			$\ordsucc_{d'}$ and $\ordvarsucc_{d'}$ share the upper contour set of (the contract $d'$ signs at) $\cop (\ordsucc_D)$.
		Therefore, $\cop (\ordsucc_D)$ is stable at $\left(\ordvarsucc_D, C_H \right)$ and $d^*$ should exist.
		Note that $d^* \neq d$ and thus $\ordsucc_{d^*} = \ordvarsucc_{d^*}$, because $d$ must be indifferent between $\cop (\ordvarsucc_D)$ and $\cop (\ordsucc_D)$ as seen above.

		Now, suppose for a moment that $x^*_\succ \succ_{d^*} x^*_\varsucc$ and let 
			\begin{align*}
				Y := \left\{ y \in \gs: \dx{y} = d^* \mbox{ and } x^*_\succ \succ_{d^*} y \right\} \ni x^*_\varsucc.
			\end{align*}
		Note that the contracts in $Y$ are never offered along the cumulative offer process at $\ordsucc_D$ even if they are acceptable.
		We thus have $ \cop\left( \orddrp{D}{Y} \right) = \cop(\ordsucc_D) $, which further leads to two observations:
		First, it is immediate to check that $\left( \orddrp{D}{Y}, \ordvardrp{D}{Y}\right)$ meet the three conditions in the statement of this lemma.
		Second, it also follows that $\cop \left( \orddrp{D}{Y} \right)  \neq \cop \left( \ordvardrp{D}{Y} \right) $ for the following reasons:
				On the one hand, $x^*_\succ \in  \cop \left( \orddrp{D}{Y} \right)$, because $x^*_\succ \in \cop( \ordsucc_D)$ by definition and $\cop (\ordsucc_D) = \cop \left( \orddrp{D}{Y} \right) $ as seen above.
				On the other hand, strategy-proofness implies $ x^*_\succ \not\in \cop \left( \ordvardrp{D}{Y} \right)$, 
						as otherwise $d^*$ can profitably manipulate by reporting $\ordvardrp{d^*}{Y}$ when the true preference is $\ordvarsucc_{d^*}$.\footnote{
						Notice that $x^*_\succ \varsucc_{d^*} x^*_\varsucc$ follows from $x^*_\succ \succ_{d^*} x^*_\varsucc$, since $\ordvarsucc_{d^*} = \ordsucc_{d^*}$ as we have mentioned above.
						}
		That is, $\left( \orddrp{D}{Y}, \ordvardrp{D}{Y}\right)$ constitutes a counterexample to the claim of this lemma.
		This, however, contradicts the minimality assumption we have imposed on $\left( \ordsucc_D, \ordvarsucc_D \right)$:
		Since $d^* \neq d$ as seen above, we have $\ordsucc_{d^*} = \ordvarsucc_{d^*} $, and hence,
			$\left| \Ac \left(\orddrp{d^*}{Y}\right) \right| = \left| \Ac \left(\ordvardrp{d^*}{Y}\right) \right| $ is strictly smaller than $| \Ac (\ordsucc_{d^*}) | = | \Ac (\ordvarsucc_{d^*})| $. 
		For any $d' \neq d^*$, $\orddrp{d'}{Y} = \ordsucc_{d'}$ and $\ordvardrp{d'}{Y} = \ordvarsucc_{d'}$.
		Thus, $(\ordsucc'_D, \ordvarsucc'_D) = \left( \orddrp{D}{Y}, \ordvardrp{D}{Y}\right)$ violates inequality (\ref{eq:lem:non-bossy}).
				
		The case of $x^*_\varsucc \succ_{d^*} x^*_\succ$ is perfectly symmetric with $Y :=\{ y \in \gs: x^*_\varsucc \succ_{d^*} y \}$, and the proof is complete.
\qed

\end{document}